\documentclass{article}

\usepackage[preprint]{neurips_2025}

\usepackage{graphicx}

\usepackage[utf8]{inputenc} 
\usepackage[T1]{fontenc}    
\usepackage{hyperref}       
\usepackage{url}            
\usepackage{booktabs}       
\usepackage{amsfonts}       
\usepackage{nicefrac}       
\usepackage{microtype}      
\usepackage{xcolor}         
\usepackage{booktabs}
\usepackage{multirow}
\usepackage{subcaption}
\usepackage{algorithm}
\usepackage{algorithmicx}
\usepackage{algpseudocode}
\usepackage{amsmath}
\usepackage{amsthm}
\usepackage{authblk}
\usepackage{xspace}
\usepackage{listings}

\newtheorem{theorem}{Theorem}
\newtheorem{definition}{Definition}

\newtheorem{example}{Example}[section]

\title{LEGO-Compiler: Enhancing Neural Compilation Through Translation Composability}

\author[1, 3]{Shuoming Zhang}
\author[1, 3]{Jiacheng Zhao}
\author[2]{Chunwei Xia}
\author[2]{Zheng Wang}
\author[1, 3]{Yunji Chen}
\author[1, 3, 4]{Xiaobing Feng}
\author[1, 3]{Huimin Cui\thanks{Corresponding Author} }
\affil[1]{SKLP, Institute of Computing Technology, CAS}
\affil[ ]{\textit {\{zhangshuoming21s,zhaojiacheng,cyj,fxb,cuihm\}@ict.ac.cn}}
\affil[2]{University of Leeds, UK}
\affil[ ]{\textit {\{C.Xia, Z.Wang5\}@leeds.ac.uk}}
\affil[3]{University of Chinese Academy of Sciences, Beijing, China}
\affil[4]{Zhongguancun Laboratory, Beijing, China}

\begin{document}

\maketitle

\begin{abstract}
Large language models (LLMs) have the potential to revolutionize how we design and implement compilers and code translation tools. However, existing LLMs struggle to handle long and complex programs.  
We introduce LEGO-Compiler, a novel neural compilation system that leverages LLMs to translate high-level languages into assembly code. 
Our approach centers on three key innovations: LEGO translation, which decomposes the input program into manageable blocks; 
breaking down the complex compilation process into smaller, simpler verifiable steps by organizing it as a verifiable LLM workflow by external tests;    
and a feedback mechanism for self-correction. Supported by formal proofs of translation composability, LEGO-Compiler demonstrates high accuracy on multiple datasets, including over 99\% on ExeBench and 97.9\% on industrial-grade AnsiBench. Additionally, LEGO-Compiler has also acheived near one order-of-magnitude improvement on compilable code size scalability. This work opens new avenues for applying LLMs to system-level tasks, complementing traditional compiler technologies.
\end{abstract}

\section{Introduction}
\label{intro}

The rapid development of Large Language Models (LLMs) has led to an expansion of their applications and effectiveness across various domains~\citep{rombach2022high,openai2023chatgpt, sora2024, copilot2024}. One important area where LLMs have shown impressive results is code translation, including tasks such as code generation from natural languages~\citep{zan-etal-2023-large} and transformation between programming languages~\citep{llmcodetranslation2024fse}. In code translation, LLMs have demonstrated remarkable accuracy and readability, often surpassing manually crafted translators.

While LLMs have shown promising results in translating between high-level programming languages~\citep{roziere2020unsupervised, roziere2021leveraging, CodeGen-ICLR23} and in decompilation tasks~\citep{decompile-fu2019coda, decompile-cao2022boosting, decompile-armengol2023slade}, their application to translating from high-level languages to low-level assembly languages remains a relatively unexplored area, which is traditionally dominated by handcrafted compilers. However, compilers require significant engineering effort and are tailored to specific languages/architectures. It is interesting to explore whether LLMs can be used to automate the compilation process, and if so, to what extent they can achieve this. Although earlier investigations~\cite{armengol2021learning,c2llvm-HPEC22} have shown low translation accuracy, recent work~\cite{llmcompiler2024emnlp} has shown that LLMs finetuned with compiler-generated bilingual corpora can outperform advanced LLMs in C-x86 compilation tasks and can achieve up to 91\% behavioral accuracy. However, an in-depth understanding of LLMs' capabilities in this domain is still lacking.
As compilation is typically divided into two main aspects: translation and optimization. This work focuses on exploring and answering about LLM capabilities in the translation aspect of compilation.

LLMs are pre-trained on vastly large code corpora. some are monolingual, and some may be bilingual (where LLMs can learn the translation rule between two languages). 
However, most of these LLMs do not disclose their training datasets, so their capabilities can only be assessed through empirical testing. 
We primarily find that current LLMs learn the neural compilation process from directly compiler-generated bilingual corpora, which is an intuitive way to construct a pretraining dataset and teach LLMs to compile.
However, we also found that assembly code directly generated by compilers is hard for LLMs to learn due to several challenges.  
These include the presence of semantically opaque labels, symbols or numeric values that LLMs struggle to translate accurately, and the need to handle symbol renaming for identifiers with the same name in different scopes, etc. 
Although style migration or modifications to existing compilers can be made, these approaches still rely on an existing compiler to perform the neural compilation job, which doesn't outperform existing designs.

Our work takes a different approach where we do not require bilingual corpora, as a result, we don't necessarily rely on an existing compiler. \textbf{We guide LLMs to transform high-level code into assembly code step-by-step.} 
To achieve this, we first propose an adaptive compilation-knowledge guided LLM workflow, which involves a series of steps with verifications to ensure stepwise correctness, including control flow annotation, struct annotation, renaming transformation, variable mapping transformation and final assembly generation. By splitting the complex compilation task into smaller, manageable steps, we significantly reduce the task complexity in each step by forcing LLMs to focus on certain easier step and achieve substantial improvements in overall accuracy.

More importantly, the scalability of current code translation is an important and challenging problem. Although advanced LLMs already have hundreds of thousands tokens context limit, they can not merely compile a code with 2.6k tokens in CoreMark~\citep{gal2012coremark}, which is just a 200-LOC function. 
The major challenges within this significant LLM failure are two-fold: \textbf{(1)} the complexity within each expression/statement, and \textbf{(2)} the complexity of program structures. \textbf{(1)} is continuingly improved with more advanced LLMs or with proper knowledge guidance and is not our research focus. However, \textbf{(2)} is more fundamentally challenging, as a program can be arbitrarily large and complex, and LLMs are not designed to handle such complexity, \textit{direct LLM translation is just not scalable.}

To tackle this scalability problem, we propose \textbf{LEGO translation}, which draws inspiration from the modular and composable nature of LEGO blocks to divide-and-conquer it. This method breaks down large programs into manageable, semantically-composable control blocks, analogous to LEGO pieces. These blocks are then independently translated and rebuilt to form a full translation in much larger scale. 

We combine the novel LEGO translation method with our proposed neural compilation workflow, and design \textbf{LEGO-Compiler}, a scalable, LLM-driven system that leverages the power of LLMs to perform scalable neural compilation tasks.
LEGO-Compiler can correctly compile over 99\% of the code in ExeBench~\citep{armengol2022exebench}, a large scale dataset with careful unit-testing. We can also correctly compile 97.9\% of AnsiBench~\cite{ansibench2011github}, a collection of well-known ANSI C standard benchmark suites, including CoreMark~\citep{gal2012coremark}, an industrial-grade codebase that encompasses most common programming language features in C, where we compile all of its 40 functions correctly.
Regarding scalability, we have verified that LEGO translation method can significantly scale up the capability of neural code translation performed by LLMs. By ablating LEGO-Compiler methods in AnsiBench evaluation and additional Csmith~\cite{csmith2011pldi} evaluation, a random C code generator for compiler testing, LEGO translation scales up the available code size for neural compilation by near an order of magnitude.

The main contributions of this work are as follows:

\begin{itemize}
\item We propose the novel LEGO translation method to scale up the neural compilation task. By breaking down large programs into manageable, semantically-composable control blocks, the complexity of neural compilation tasks for LLMs is significantly reduced.
\item We propose a novel verifiable step-by-step neural compilation workflow that guides LLMs to transform high-level code into low-level assembly. With breakdown steps, we characterize and evaluate the compilation process in LLM's perspective, and achieve substantial improvements in behavioral accuracy compared to end-to-end translation.
\item We provide both theoretical and empirical studies by formally defining the composability in code translation that underpins the LEGO translation method and empirically demonstrating LEGO-Compiler's effectiveness through extensive evaluations. LEGO-Compiler achieve over 99\% accuracy on ExeBench, 97.9\% accuracy on AnsiBench. Ablation study also showcases that LEGO translation boosts the scalability of neural-compilable code size by an order of magnitude. The model-independent evaluation process also serves as an important benchmark for LLMs on complex system-level tasks.
\end{itemize}
\section{Related Work}
\label{related}

\subsection{Code Translation}
Recent neural-based \textbf{Code Translation} researches can be majorly categorized to two types: learning-based transpilers~\citep{roziere2020unsupervised, roziere2021leveraging, wen2022babeltower} and pre-trained language models~\citep{feng2020codebert, wang2021codet5, lu2021codexglue, roziere2022code, openai2023gpt4, Claude_3}.
The former majorly studies the scarcity of parallel corpora~\citep{xie-etal-2023-data} and develops unsupervised learning methods to overcome it. The latter using Large Language Models' vast pretrained knowledge, can also perform code translations well without training~\citep{llmcodetranslation2024fse, interactivecodetranslation2024tse}.

As for \textbf{compilation related translations}, ~\cite{armengol2021learning, c2llvm-HPEC22} preliminarily study on C-x86 and C-LLVM IR translation with limited investigations on the methods. There are also works on the reverse decompilation process~\citep{decompile-fu2019coda, decompile-cao2022boosting, decompile-armengol2023slade} and works on code optimizations~\citep{cummins2023large, cummins2024metalargelanguagemodel}. The most related work is~\cite{llmcompiler2024emnlp}, which achieves state-of-the-art 91\% Pass@1 accuracy on the C-x86 task using a finetuned CodeLlama model, where our work surpasses. Besides, their approach relies on compiler-generated bilingual corpora, while our methods can effectively eliminate such dependency by reasoning the steps of how a compiler works.

Finally, \textbf{Modular approach} is recognized as the key insight to scale up neural code generation/translation~\cite{divideconquerSMT2015, parsel2023neurips, planningcodegeneration2024arixv, divideanconquer2024neurips}.
Our work leverages similar idea of divide-and-conquer to breakdown a large long code into manageable control block parts, then LLMs can translate these parts separately with the aid of necessary context and combine their results into a large, complete and coherent translation.

\subsection{Other Related Work}
\textbf{LLM self-repair.} Recent research has focused extensively on enhancing LLMs' self-correction capabilities. Several studies closely related to our work deserve mention. A comprehensive survey by ~\cite{pan-etal-2024-automatically} thoroughly examined methods for leveraging feedback to autonomously improve LLM outputs. 
~\cite{wang-etal-2022-compilable} first uses compiler feedback for better code generation, and ~\cite{dou2024whatswrongcodegenerated} establishes the syntax-runtime-functional bug type taxonomy and builds corresponding self-repair pipelines for code. Our work is their natural extensions to neural compilation scenario. While ~\cite{olausson2024is} investigated the limitations of self-repair mechanisms in code generation, our findings diverge significantly. Contrary to their conclusions, we discovered that self-repair serves as a highly effective solution in the neural compilation process, particularly when incorporating syntax feedback and runtime feedback.  

\textbf{In-context learning and Chain-of-Thoughts.} LLMs are able to in-context learn via inference alone by providing few shots of demonstration then predicting on new inputs~\citep{min-etal-2022-rethinking, dong2024surveyincontextlearning}. Thus customized Chain-of-Thoughts~\citep{NEURIPS2022COT, chu2024navigateenigmaticlabyrinthsurvey} can guide LLMs to perform complicated reasoning~\citep{wang-etal-2022-iteratively, song2024manyshotincontextlearninghelp}, which is the cornerstone of our work. More specifically, ~\cite{levy-etal-2024-task} reveals the degradation of LLMs' performance for long context, and validate the effectiveness on using Chain-of-Thoughts to mitigate.
We found similar results in code translation/compilation tasks. However, our proposed \textbf{LEGO translations} method can significantly mitigate such degradation as it turns a long context direct translation into multiple composable, shorter ones that LLMs can handle.

\textbf{Generation Scalability and Long Context Learning.} Except for code translation, many LLM-based methods suffer scalability problems since larger inputs are not well trained like the smaller ones, which makes general methods to extend LLMs long-context capability challenging. For example, in order to coherently generate long passages of text, ~\cite{tan-etal-2021-progressive} proposes a multi-staged keyword-first progressive method to improve it significantly, where our work shares a similar insight.
~\cite{li2024retrievalaugmentedgenerationlongcontext} introduces a self-route method to dynamically choose the usage of RAG or fully in-context, balancing the cost and performance in long-context scenario, which inspires us to use a similar dynamic approach.
\section{Methods}
\label{method}

\subsection{Problem Definition}
Before introducing our method, we first define the neural compilation problem. Neural compilation can be viewed as a specialized version of code translation problem, as defined in \autoref{def:code_translation}, with the goal of translating high-level programming language as the \textit{src} language (such as C) into low-level assembly language as the \textit{dst} language (such as x86, ARM, or RISC-V). Unlike general code translation, compilation needs to handle more low-level details, such as memory layout and calling convention, while ensuring the functional correctness of the translated result. 

\begin{figure}
\centering
\includegraphics[width=0.98\textwidth]{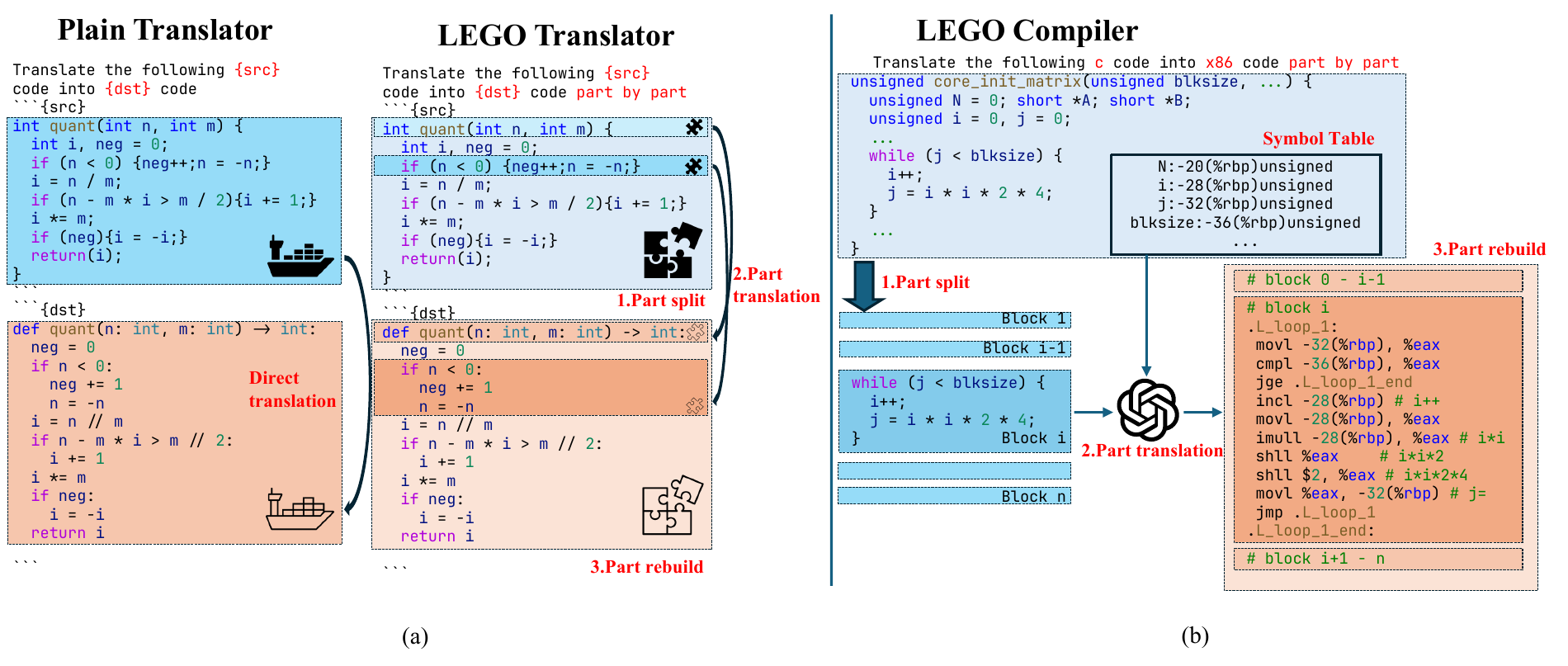}
\caption{\textbf{a.} Plain translation vs LEGO translation, by splitting the program into smaller composable control blocks(parts), translating each part becomes an easier task, and rebuilding each translated partial result will form a full translation. \textbf{b.} LEGO compiler, a special case for LEGO translation, to translate each part correctly, a symbol table needs to be maintained first and provided during translation.}
\label{insight_fig}
\end{figure}

\begin{definition}
\label{def:code_translation}
    There are two programming languages: $\mathcal{L}_{src}$ and  $\mathcal{L}_{dst}$, each is an infinite set of valid program strings. There exists a unary relation $\rightharpoonup$ from $\mathcal{L}_{src}$ to $\mathcal{L}_{dst}$. The problem is to perform a translator function $T$: $\forall x \in \mathcal{L}_{src}, (\exists u \in \mathcal{L}_{dst}, x \rightharpoonup u) \rightarrow (x \rightharpoonup T(x)), T(x) \equiv x $ semantically.
\end{definition}

\subsection{\textbf{LEGO Translation}: Core Method}

As depicted in (a) in \autoref{insight_fig}, previous neural code translation methods typically convert entire programs at the function or file level. While this approach may be effective for smaller programs, it struggles with larger programs due to significant accuracy degradation. These methods translate code at a coarse granularity, making it challenging to translate very long functions using LLMs. 
This limitation is stark: taking neural compilation as an example, even state-of-the-art LLMs~\cite{Claude_3,openai2024gpt4o}, despite possessing context windows potentially spanning hundreds of thousands of tokens (e.g., 128k-200k), demonstrably fail to correctly compile a C function exceeding just 2.6k tokens using direct translation.
They could also perform code-snippet level translation, but they lack guidelines and necessary information to compose the code-snippet level results together, and there is also no clear formal proof of the composability of code.
Despite these limitations, we observe an inherently composable nature in code. In the context of neural compilation, we propose the following insights to enhance translation scalability:
\begin{itemize}
    \item Fine-grained translation: Instead of translating an entire program at once, focus on translating smaller code snippets accurately. By ensuring each part is correctly translated, they can be combined to form a semantically equivalent complete translation.
    \item Contextual Awareness: Effective translation of smaller code snippets requires understanding their contextual positioning within the code. This includes recognizing the relationship with preceding and succeeding snippets to maintain semantic coherence.
    \item Symbol Handling: Accurate translation necessitates careful management of symbols (like variable scopes, types, and memory locations) and program constructs within each block to ensure correct mapping to the target architecture's semantics and preserve functionality.
\end{itemize}

Inspired by ~\cite{wang2024legoprover}, where this process is similar to the destruct and rebuild process of a LEGO toy, we named the fine-grained translation technique as \textbf{LEGO translation} and our system built upon it as \textbf{LEGO-Compiler}. 
As depicted in (b) in \autoref{insight_fig}, LEGO translation first breaks down large programs into manageable, self-contained blocks analogous to LEGO pieces (\textbf{Part split}). Then these blocks are independently translated (\textbf{Part translation}) and finally recombined, enabling scalable and accurate translation of complex programs (\textbf{Part rebuild}). All these methods rely on an inherent nature in programming languages, the composability in control block level, which reflects the linearization process in compiler design~\citep{wirth1996compiler}, where tree-structured control flow can be linearized, and therefore, composable. We have proved the widely applicable composability of programming languages using a constructive approach in \autoref{proof_composable}.

\subsection{A Verifiable, Stepwise Neural Compilation Workflow}
\label{verifiable_workflow}

Directly translating complex high-level code to low-level assembly poses significant challenges for LLMs. To address this, we structure the neural compilation process as a \textbf{stepwise workflow}. This approach decomposes the overall task into a sequence of distinct, more manageable sub-tasks, each designed to be handled effectively by an LLM guided by specific prompts and context.

A key principle integrated throughout this workflow is \textbf{verifiability}. As depicted in ~\autoref{fig_cot_example}, by breaking down the complex, hard-to-verify compilation process into smaller steps, we create opportunities to validate the intermediate results of many stages before proceeding. This significantly enhances the reliability of the entire process and contributes to understanding the capabilities and limitations of LLMs in compilation tasks. Verification techniques employed vary depending on certain step and may include static source code analysis (e.g., comparing ASTs), cross-referencing intermediate calculations against compilation-based frontend tools, ensuring behavioral equivalence through execution testing, and potentially applying formal methods like SMT-based checks for specific properties like memory safety.

\begin{figure}
\centering
\begin{minipage}{0.48\textwidth}
  \centering
  \includegraphics[width=\textwidth]{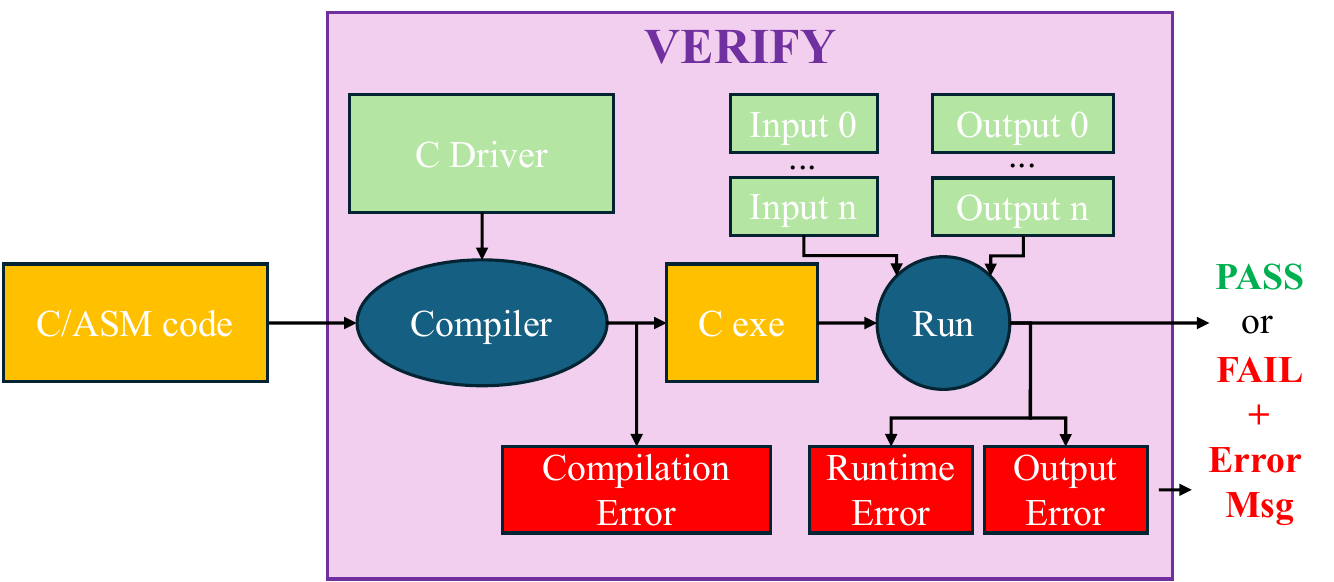}
\end{minipage}
\hspace{0.05\textwidth}
\begin{minipage}{0.42\textwidth}
  \centering
  \includegraphics[width=\textwidth]{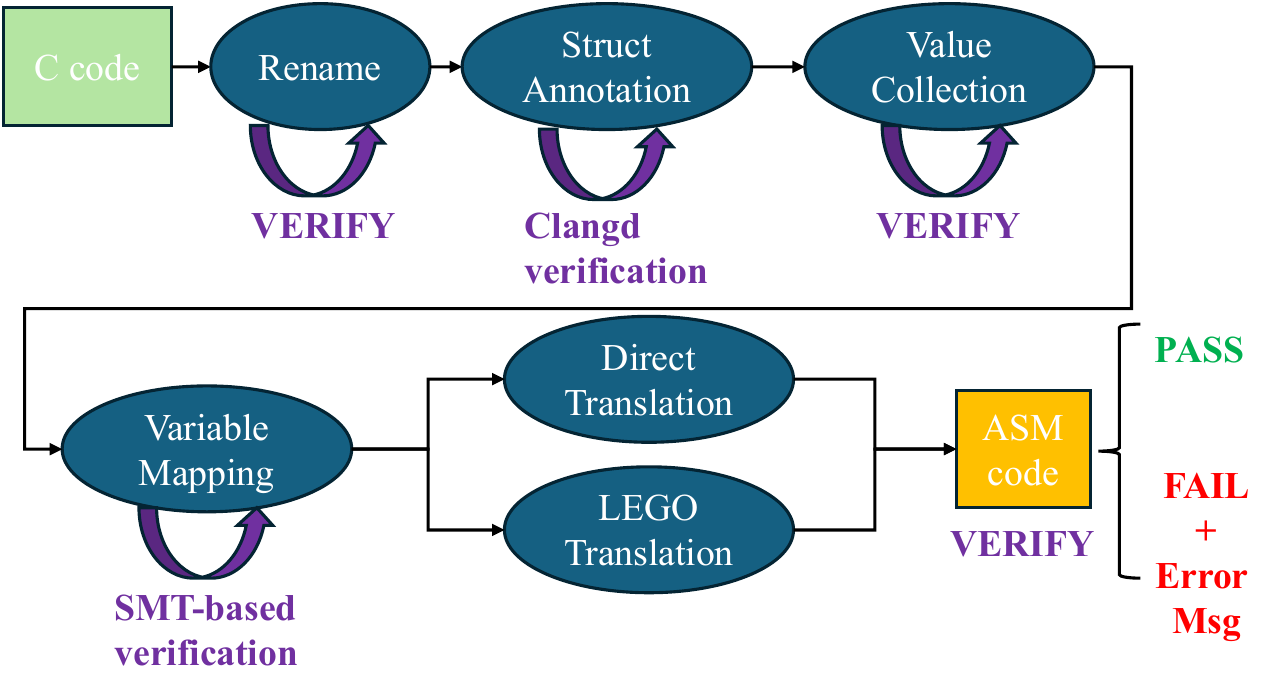}
\end{minipage}
\caption{Neural compilation workflow in \textbf{LEGO-Compiler}. Left figure shows the behavioral verification process with unit-tests. Right figure shows the detailed steps in the workflow, some step is \textit{residual} that may be skipped as performing such step may be unnecessary for certain input program.}
\label{fig_cot_example}
\end{figure}

As depicted in ~\autoref{fig_cot_example}, the workflow proceeds through the following major steps:
\begin{itemize}
  \item \textbf{Variable Renaming:} An initial source-level transformation ensures all variable identifiers within the compilation scope (e.g., a function) have unique names, resolving potential ambiguities from name shadowing and simplifying subsequent mapping. \textit{This renaming step is verifiable} by  executing the original and renamed source code on test cases to ensure behavioral equivalence. 
  \item \textbf{Type and Layout Analysis:} This stage focuses on understanding the program's data structures. The LLM performs a structured reasoning process for compound types (structs, unions, arrays) to determine their memory layout (size, alignment, member offsets) based on target architecture conventions and constituent basic types. \textit{The correctness of this analysis is verifiable} by cross-referencing the inferred type sizes and offsets against outputs from standard development tools like Clangd~\citep{clangd2024} or IntelliSense~\citep{intellisense2024}.
  \item \textbf{Variable Mapping and Allocation:} This step identifies all variable instances and determines the correspondence between high-level variables and their low-level assembly representations. Global variables are mapped to labeled memory, while local variables are assigned stack offsets relative to a base pointer. Access to compound type elements uses calculated offsets adhering to conventions like the System V ABI~\citep{SystemVABI2018}. \textit{Verification for this stage can involve checks} using techniques like SMT-based verification~\cite{z3solver2008} to detect potential memory allocation issues (e.g., overlaps, out-of-bounds accesses) based on the derived allocation plan.
  \item \textbf{Part Split (Control Flow Decomposition):} Leveraging the LEGO translation principle, this step decomposes the input function into smaller, manageable control blocks. It analyzes the program's control flow graph (CFG) and uses an LLM-driven process to perform an adaptive algorithm in ~\autoref{algo_part_split} to decide where to split, aiming for semantically coherent units suitable for independent translation. \textit{The structural integrity of the split is verifiable} by ensuring that the CFG formed by recombining the split blocks (before translation) is isomorphic to the original function's CFG, or more simply, immediate recombination.
  \item \textbf{LEGO Part Translation:} Each control block generated by the Part Split step is translated independently by the LLM into assembly code. The LLM receives the source code for the block along with relevant context derived from previous steps, such as the established variable mappings and type layouts.
  \item \textbf{Part Rebuild and Final Verification:} The translated assembly blocks are reassembled according to the original control flow structure, typically, for two adjacent code blocks, the assembling operation is just concatenation as assembly language is linearized. \textit{The functional correctness of the final, combined assembly code generated by the entire workflow (Split, Translate, Rebuild) is verified} through behavioral equivalence checks against the original source code, implemented using unit-tests.

\end{itemize}

Finally, LEGO-Compiler integrates a \textbf{self-correction loop with error feedback} as a final quality check after full translation. This mechanism detects residual errors using the assembler (semantic errors), runtime execution/debuggers (runtime errors), and behavioral testing via unit tests (behavioral errors). Diagnostic information is fed back to the LLM to iteratively refine the generated assembly. This self-correction process is crucial for enhancing the robustness and accuracy of the LLM-based compilation system. As LLMs are non-deterministic and may exhibit trivial errors, where the feedback loop is essential for correcting these errors.

\section{Experiments}
\label{experiment}

\subsection{Experimental Setup}
Major parameters we have tested are listed below:
note that not all combinations of experimental settings are tested due to resource constraints.
\begin{itemize}
    \item \textbf{Models}: We select a variety of state-of-the-art LLMs from different vendors, including OpenAI's newest GPT-4.1 and its mini version~\cite{openai2025gpt41}, Anthropic's Claude-3.5-sonnet~\cite{Claude_3} and Claude-3.7-sonnet~\cite{anthropic2025claude37}, Deepseek's Deepseek-V3 and its newest 0324 version~\cite{deepseekai2024deepseekv3technicalreport}, and Google's Gemini-2.0-flash and Gemini-2.5-pro~\cite{gemini2024arxiv}. We select models pairwisely to illustrate the model-side improvement on the neural compilation task.
    \item \textbf{Benchmarks}: We majorly test on ExeBench~\cite{armengol2022exebench}, a large-scale dataset of executable C programs, additionally we use AnsiBench~\cite{ansibench2011github} and Csmith-generated programs~\cite{csmith2011pldi} as case studies. Technically, we use ExeBench's Real-Executable subset, initially containing over 40k cases, after data cleaning and removing cases uncompilable by an oracle compiler, we finally obtain a 17,121 cases testset of ExeBench, which is too large for full scale evaluation. We further filter two subsets of ExeBench for evaluation, we filter a hard-cases subset of 1,996 samples based on the number of basic blocks and instructions within these blocks using the LLVM toolchain~\cite{LLVM2024}, shown in \autoref{fig:exebench_breakdown_grid_view}, and the other is a randomly selected subset with 2000 cases, which is used for comparison and ablation studies.
    \item \textbf{T}emperature: 0.0-1.0, with 0.2 step increments
    \item \textbf{Architecture}: \textbf{x86}\_64, \textbf{arm}-v8a, \textbf{riscv}64, majorly on \textbf{x86}
\end{itemize}

\subsection{ExeBench Evaluation}

\begin{table}[t]
  \caption{ExeBench (17,121 cases) experimental results with method-level ablation on C-x86 neural compilation task compared to previous state-of-the-art~\cite{llmcompiler2024emnlp}, the base model performs similarly, but with our proposed methods, we achieve substantial accuracy improvement on the full dataset. 
  Except DS-V3-0324 model, other models are evaluated on a subset of 2000 cases due to budget constraints, DS-V3-0324 is evaluated on both datasets and performs nearly identical (<0.2\%). 
  }
  \label{tab:results_ablation}

  \centering
  \begin{tabular}{lccc}
    \toprule
    Model & Direct & CoT workflow & LEGO translation \\
    \midrule
    DS-V3-0324-full & 91.589\% & 94.708\% & 99.375\% \\
    Zhang et.al~\cite{llmcompiler2024emnlp} & 91.718\% & - & - \\
    \midrule
    \midrule
    DS-V3-0324 & 91.60\% & 94.60\% & 99.20\% \\
    DS-V3-1226 & 88.7\% & 94.55\% & 98.30\% \\
    \midrule
    GPT-4.1 & 90.20\% & 95.10\% & 99.50\% \\
    GPT-4.1-mini & 78.90\% & 91.30\% & 97.70\% \\
    \midrule
    Claude-3.7-sonnet & 95.20\% & 97.75\% & 99.70\% \\
    Claude-3.5-sonnet & 92.40\% & 95.25\% & 99.15\% \\
    \midrule
    Gemini-2.5-pro & 94.45\% &97.25\% & 99.65\% \\
    Gemini-2.0-flash & 73.10\% & 84.60\% & 92.60\% \\
    \bottomrule
  \end{tabular}
\end{table}

\begin{table}
  \caption{ExeBench Hard Subset (1996 cases) evaluation: The applied filters are based on nums of basic blocks(10), and max(80)/all(200) instructions within these blocks analyzed by LLVM toolchain~\cite{LLVM2024}, Detailed characterization of ExeBench and its hard subset is on ~\autoref{fig:exebench_breakdown_grid_view}.}
  \label{tab:results_hard_exebench_ablation}
  \centering
  \begin{tabular}{lccc}
    \toprule
    Model & Direct & CoT workflow & LEGO translation \\
    \midrule
    GPT-4.1 & 87.43\% & 94.54\% & 98.70\% \\
    Claude-3.7-sonnet & 92.59\% & 97.14\% & 99.20\% \\
    Deepseek-V3-0324 & 85.92\% & 91.38\% & 97.60\% \\
    Gemini-2.5-pro & 94.29\% & 97.09\% & 99.10\% \\
    \bottomrule
  \end{tabular}
\end{table}

Recall ~\autoref{fig_cot_example}, we evaluate ExeBench through the following setup:
\begin{enumerate}
     \item Translate the C program to assembly (to generate hypothesis), where we have three different methods: direct translation(as baseline), stepwise workflow translation and LEGO translation. 
    Note that some steps in the workflow is necessary for the LEGO translation method, as a global context is needed.
    \item Assemble and link the hypothesis assembly to create an executable.
    \item Run the executable through 10 different IO test cases provided by ExeBench.
    \item Consider the translation \textit{successful} if it passes all test cases.
    \item If a translation fails (at any former step), apply self-fixing with the collected error feedback to the LLM, will try \textbf{k} rounds. We set \textbf{k} to 5 during the evaluation.
    \item Consider the translation \textit{failed} if it doesn't pass after all configured attempts.
\end{enumerate}

~\autoref{tab:results_ablation} and ~\autoref{tab:results_hard_exebench_ablation} summarize the empirical results of our LEGO-Compiler on ExeBench targeting x86-64 architecture. We establish carefully crafted 1-shot prompts to guide LLMs for each translation method.
With our proposed methods, all models achieve substantial improvements, where newest models' accuracy on the whole dataset/hard subset reaches averagely 99.56\% and 98.65\%.
Claude-3.7-sonnet and Gemini-2.5-pro models are the best performing models during our evaluation, achieving over 99\% accuracy on the hard subset.
We analyze the ablated results as follows:
\begin{itemize}
  \item \textbf{Step-by-step workflow} improves the translation majorly related to complex data structures and many variable assignments, where direct translation may fails to handle such complexity all together.
  \item \textbf{LEGO translation} majorly improves lengthy code translation with multiple control statements, where the model struggles to keep track of the context and the correct label usage without our divide-and-conquer methodology.
  \item \textbf{Self-correction} fixes most trivial errors related to architecure-specific knowledge, and improves at all methods as it is orthogonal. Two major types of observed errors are: 1) misuse of instruction operands, like `cmp' instructions cannot compare two immediate values or two memory values; 2) mnemonics-related, like access global variables or values stored in data section, LLMs need to generated \%rip relative addressing operands instead of direct label usage. Taking Deepseek-V3-0324 as an example, 1) and 2) account for 26 and 40 failed cases in its 172 failed cases during CoT workflow evaluation.
\end{itemize}

\begin{table}
\centering
  \caption{ExeBench's hard subset evaluations with different architectures with Gemini-2.5-pro.}
  \label{tab:results_exebench_arm_riscv}
  \begin{tabular}{lcccc}
    \toprule
    Architecture & Model & Direct & CoT workflow & LEGO translation \\
    \midrule
    x86\_64 & Gemini-2.5-pro & 94.29\% & 97.09\% & 99.10\% \\
    arm64 & Gemini-2.5-pro & 87.64\% & 92.33\% & 96.74\% \\
    riscv64 & Gemini-2.5-pro & 84.22\% & 89.88\% & 94.64\% \\
    \bottomrule
  \end{tabular}
\end{table}

Except x86 evaluation, we also evaluate LEGO-Compiler on arm64 natively on Apple M1 chip and riscv64 through Spike simulator. Due to time and budget constraints, we only evaluate the hard subset as it contains more challenging cases. As depicted in ~\autoref{tab:results_exebench_arm_riscv}, the evaluation results are similar to x86, where our LEGO-Compiler powered by Gemini-2.5-pro achieves similar improvements with our proposed methods. The globally lower accuracy may due to lower pretrained knowledge of these assembly languages in LLMs, or insufficient prompt engineering efforts as we have not thoroughly tested prompts usage on these architectures. It is also an interesting work to automate the prompt engineering process for different architectures to inject compilation-related knowledge to the translation process, which we leave for future work.

Another finding is observed through pairwise comparison of models, where we find clear improvement of newer/larger models over older/smaller models. We analyze the reasons as two-fold: First, more advanced new models are pretrained with more compilation-related knowledge, which helps the translation of certain expressions and statements. Second, newer models are more capable of reasoning, which is critical for the workflow translation and LEGO translation methods.

To sum up, the empirical results of our LEGO-Compiler system is promising, we prove a training-free approach to use LLMs as neural compilers, which can successfully translate averagely 99.56\% of ExeBench testset and 98.65\% of its hard subset across advanced LLMs from 4 state-of-the-art vendors. The model-independent evaluation process also establishes a challenging benchmark for LLMs, which requires 3 key capabilities: 1) mathematical reasoning and long-context reasoning 2) code/assembly understanding and translation, 3) error localization and correction.

\subsection{AnsiBench: more real-world codebase evaluation}

We conduct additional real-world codebases evaluation, we use AnsiBench~\citep{ansibench2011github}, a collection of well-known ANSI C standard benchmark suites~\citep{gustafson1995hint,dongarra2003linpack,gal2012coremark}, benchmarking a wide variety of systems and compilers, including a number of classic, industry-standard benchmarks as well as some selected programs that can be used as benchmarks.

We evaluate the whole AnsiBench collection with our LEGO-Compiler, powered by Claude-3.7-Sonnet, our best-performing model in previous evaluation. We list the details of every function we compiled in \autoref{fig_ansibench_1}, totally we have 96 functions in total, except for few utility functions which are easy to compile, most of them represent real-world codebase complexity. We ablate the translation methods we applied to showcase both the effectiveness of CoT-like workflow and LEGO translation.
In total, we pass 94 out of 96 cases in Ansibench across 7 different codebases, including Whetstone, Dhrystone, Hint(one failure), Linpack, Tripforce(one failure), Stream and CoreMark.When measured by token count, LEGO translation method significantly improves the translation scalability of real-world code by near an order of magnitude as illustrated in ~\autoref{fig_ansibench_1}.

\begin{figure}
  \begin{center}
  \includegraphics[width=0.5\textwidth]{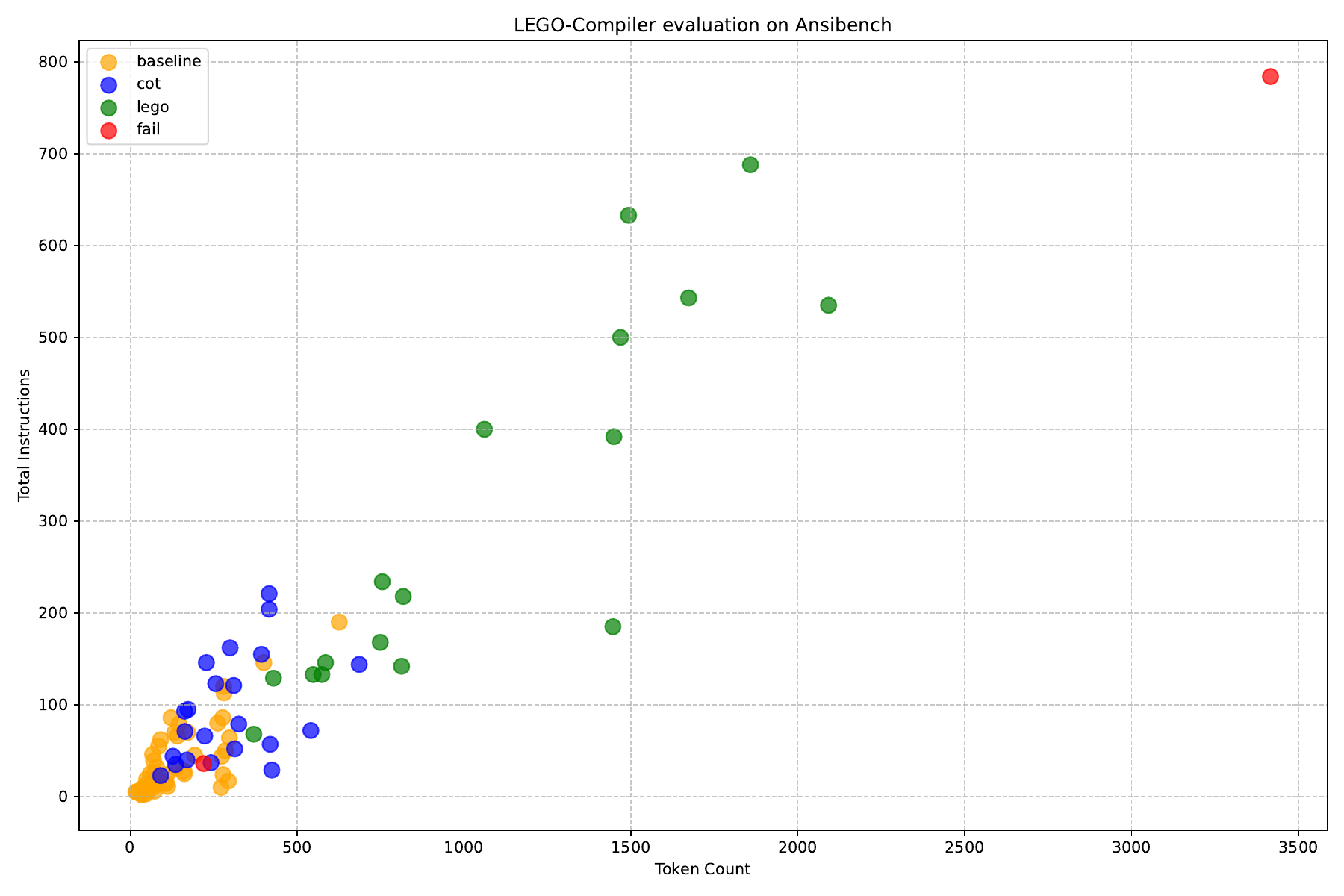}
  \end{center}
  \caption{AnsiBench evaluation results with Claude-3.7-Sonnet. The \textbf{token count} only computes the input length of C code, and typically, the output assembly will be 3-6 times larger in token size.}
  \label{fig_ansibench_1}
\end{figure}

There are majorly three types of errors where the first two types are where LEGO translation is superior. 
\begin{itemize}
    \item Lengthy code input with over a thousand token size (typically), where the output size is truncated due to limited model output length. Besides, the coarse-grained translation itself is prone to bugs. LEGO translation method can significantly reduce such errors, the case in which LEGO translation also fails is the main function of Hint benchmark, which is even more complex than the main function of CoreMark depicted in ~\autoref{coremark_main}. We analyze its failure, where the LLM-reasoning step of the stack allocation fails to generate a correct mapping. Despite this, LEGO translation handles all the other lengthy code correctly as it successfully reduces the translation complexity to control-block level.
    \item Long context forgetting problem~\cite{longcontext2024tacl}, where the model can not match the current processing assembly with the source code faithfully, LEGO translation method, on the other hand, can handle these cases efficiently with less unnecessary contexts that may cause these `random' errors. Besides, finer-grained translation also gives LLMs more attention to faithful translation of operations, the order of operations and implicit conversions. 
    \item Insufficient pretraining in LLMs, where LLMs lack certain knowledge to perform certain expression/statement translation or other architecture-specific details. For example, the other error in AnsiBench, the \texttt{generate\_password} function in TripForce, where the translation fails to translate the multiline strings correctly.
    Feedback correction can mitigate such failures. Besides, a clear model-level improvement is observed by all model pairs, and we can be positive about these failures because as LLMs advance with more pretrained knowledge, their peformance will improve as well.
\end{itemize}

\subsection{Csmith: randomly generated programs evaluation}

Except for AnsiBench evaluation. We further perform evaluations on randomly generate programs with sufficient complexity.
We use Csmith~\citep{csmith2011pldi}, a random generator of C programs which is widely used for finding compiler bugs using differential testing as the test oracle. Typically, Csmith examines compilers with random programs with corner case features and numbers, testing the robustness of compilers. Code examples generated from Csmith are illustrated in \autoref{fig_csmith_example1}.  

We follow similar ablation strategy in AnsiBench evaluation.
As depicted in \autoref{fig_csmith_results}, randomly generated programs by Csmith are very hard for both baseline and CoT-only methods to translate. 
In a randomly generated test suite of 40 cases generated by Csmith, LEGO translation successfully compiles 25 cases, while baseline translation can only compiles 4 cases, and 13 cases for CoT workflow. Besides, the complexity of cases passed by LEGO translation method are significantly larger than others, characterized by token count, basic block count and total instructions, LEGO translation scales in code size and complexity by near an order of magnitude.

During Csmith evaluation, we also identify several kinds of errors during LEGO-Compiler translation. For example, overflow value assignment is an error which doesn't occur usually but can be found in compiler testing. Taking \texttt{int16\_t x = 0x56671485;} as an example, it will trigger errors because LLMs directly generate \texttt{movw \$0x56671485, x's address} in x86, which fails to check whether the numerical value (overflows the 16 bit word) can be represented through \texttt{movw} instruction. Another example is, when handling with implicit type conversions, LLMs may not cast the type correctly, this is critical for floating point computation as operations with wrong precision will cause accumulated numerical errors. As a result, LEGO-Compiler in Csmith evaluation only achieves moderate behavioral accuracy of 62.5\%. 

However, Csmith-generated test cases do not commonly appear in real-world usages. Therefore, LEGO-Compiler is still promising in compiling common programs as evaluated by ExeBench and AnsiBench evaluation, indicating great potentials in the field of neural compilations. 
\section{Conclusion}
\label{conclusion}

We have presented LEGO-Compiler, a novel approach to neural compilation that leverages Large Language Models (LLMs) to translate high-level programming languages into assembly code. Our LEGO translation method breaks down large programs into manageable, self-contained blocks through the composable nature of code, significantly extending the scalability of neural code translation. By incorporating a series of Chain-of-Thought stages guided by classical compiler design and self-correction mechanisms, LEGO-Compiler effectively addresses key challenges in compilation tasks, achieving significant improvements in accuracy and scalability, with both theoretical and empirical studies demosntrating the effectiveness of our approach. 

These findings provide important insights into the capabilities and limitations of LLMs in neural compilation tasks. As LLM capabilities continue to improve, approaches like LEGO-Compiler are poised to play an increasingly important role in the future of software development and compilation, complementing and enhancing traditional compiler technologies.

\bibliography{main.bib}
\bibliographystyle{plain}

\newpage
\appendix

\section{Composability of C-like Language Constructs}
\label{proof_composable}
\subsection{Definitions and Language Structure}

We define a simplified C-like language structure using the following EBNF-inspired grammar:

\begin{verbatim}
block: '{' (blockItem)* '}';
blockItem: decl | stmt;
stmt:
    lVal '=' exp ';'                            # assignStmt
    | exp ';'                                   # exprStmt
    | 'goto' label ';'                          # gotoStmt
    | ';'                                       # blankStmt
    | block                                     # blockStmt
    | IF '(' exp ')' stmt (ELSE stmt)?          # ifStmt
    | WHILE '(' exp ')' stmt                    # whileStmt
    | FOR '(' stmt exp ';' stmt ')' stmt        # forStmt
    | SWITCH '(' stmt ')' stmt                  # switchStmt
    | BREAK ';'                                 # breakStmt
    | CONTINUE ';'                              # continueStmt
    | RETURN (exp)? ';'                         # returnStmt;
\end{verbatim}

We derived from the grammar that describes C-like language to form the following definitions. Also for simplicity purposes, we omit the slight differences between \textbf{decl}, \textbf{stmt} and \textbf{exp}.

\begin{definition}[Basic Statement]
\label{def:basic_statement}
A basic statement is a statement that does not contain any other statements within its structure. This includes assignStmt, exprStmt, gotoStmt, blankStmt, breakStmt, continueStmt, and returnStmt. We first exclude gotoStmt for the main proof for simplicity.
\end{definition}

\begin{definition}[Basic Block]
\label{def:basic_block}
A basic block is a sequence of consecutive basic statements as defined in \autoref{def:basic_statement}, in which flow of control enters at the beginning and leaves at the end without halt or possibility of branching except at the end.
\end{definition}

\begin{definition}[Control Block]
\label{def:control_block}
A control block is a code snippet that reflects a complete control structure, such as for(;;)\{\}, if()\{\}[else\{\}], while()\{\}, do\{\}while(), or switch()\{case:...\}. Each subpart of a control block can be other control blocks or basic blocks as defined in \autoref{def:basic_block}.
\end{definition}

\begin{definition}
\label{def:basic_control_block}
A basic control block is an innermost control block (\autoref{def:control_block}) where each of its subparts contains only basic blocks as defined in \autoref{def:basic_block}.
\end{definition}

\begin{definition}[Compound Control Block]
\label{def:compound_control_block}
A compound control block is a control block (\autoref{def:control_block}) that contains at least one subpart that is not a basic block (\autoref{def:basic_block}), but rather another control block as defined in \autoref{def:control_block}.
\end{definition}

\begin{definition}[Translation Function and Valid Translations]
\label{def:translation_function}
Let $\mathcal{T}$ be the set of all valid translation functions from $SRC$ to $DST$, where $SRC$ is the source language (our C-like language) and $DST$ is the destination language (e.g., x86 assembly). 

Formally, $\mathcal{T} = \{T \mid T: SRC \rightarrow DST\}$ such that for any $T \in \mathcal{T}$ and any $stmt \in SRC$:

1. $T(stmt) \in DST$
2. $T(stmt)$ preserves the semantics of $stmt$

A translation function $T \in \mathcal{T}$ maps each construct in the source language to one or more constructs in the destination language while preserving the program's behavior.
\end{definition}

\begin{definition}[Translation Composability]
\label{def:translation_composability}
Let $(SRC, \circ)$ be the source language with concatenation operation $\circ$, and $(DST, \cdot)$ be the destination language with concatenation operation $\cdot$. Let $\mathcal{T}$ be the set of valid translation functions as defined in \autoref{def:translation_function}. 

Translation composability holds if and only if:

\begin{equation*}
\exists T \in \mathcal{T} : \forall P_1, P_2 \in SRC, T(P_1 \circ P_2) \equiv T(P_1) \cdot T(P_2)
\end{equation*}

Where:
\begin{itemize}
    \item $T: SRC \rightarrow DST$ is a translation function
    \item $\equiv$ denotes semantic equivalence, preserving both control flow and data flow
    \item $\circ: SRC \times SRC \rightarrow SRC$ is the concatenation operation in the source language
    \item $\cdot: DST \times DST \rightarrow DST$ is the concatenation operation in the destination language
\end{itemize}
\end{definition}

\subsection{Composability of Basic Statements}

\begin{theorem}[Composability of Basic Statements]
\label{thm:basic_statement_composability}

For any two basic statements $stmt_1$ and $stmt_2$ in SRC, as defined in \autoref{def:basic_statement}, their translation is composable:
$T(stmt_1 \circ stmt_2) \equiv T(stmt_1) \cdot T(stmt_2)$
\end{theorem}

\begin{proof}
We prove this for all combinations of assignment statements and expression statements. The proof considers control flow preservation, data flow preservation, and independence of translation. Other basic statements (blank, return, etc.) trivially maintain composability as they do not affect control or data flow when composed with other basic statements.
\end{proof}

\begin{example}
\label{ex:basic_statement_composability}
This example illustrates the composability of basic statements as defined in \autoref{def:basic_statement} and proved in \autoref{thm:basic_statement_composability}.

Consider the following sequence of basic statements:

\begin{verbatim}
a = b + 3;  // stmt_1
b = a - 1;  // stmt_2
\end{verbatim}

The translation of these statements might look like:

\begin{verbatim}
T(stmt_1):
    mov eax, [b]
    add eax, 3
    mov [a], eax

T(stmt_2):
    mov eax, [a]
    sub eax, 1
    mov [b], eax
\end{verbatim}

These translations are composable because:

1. Control Flow: The order of execution is preserved (stmt\_1 then stmt\_2).
2. Data Flow: The value of 'a' computed in stmt\_1 is correctly used in stmt\_2.
3. Independence: The translation of stmt\_2 does not depend on how stmt\_1 was translated, only on its effect (the value of 'a').

Therefore, $T(stmt_1 \circ stmt_2) \equiv T(stmt_1) \cdot T(stmt_2)$, demonstrating composability.
\end{example}

\autoref{ex:basic_statement_composability} illustrates that even when statements have data dependencies, their translations remain composable as long as the order of operations is preserved. Similar proof of composability can be made for all stmts within a basic block (\autoref{def:basic_block}).

\subsection{Composability of Basic Control Structures}

\begin{theorem}[Composability of Basic Control Structures]
\label{thm:basic_control_structure_composability}
Basic control structures (if-else, for, while, do-while, switch-case), where all their components are basic blocks as defined in \autoref{def:basic_block}, are composable under the translation function $T$ as defined in \autoref{def:translation_function}.
\end{theorem}

\begin{proof}
We will prove this for each basic control structure:

1. For Loop:

Let $B_{init}$, $B_{cond}$, $B_{incr}$, and $B_{body}$ be the basic blocks for init, cond, incr, and body respectively.

Translation structure:
\begin{verbatim}
T(basic_for_loop):
    T(B_init)
loop_start:
    T(B_cond)
    jz loop_end
    T(B_body)
    T(B_incr)
    jmp loop_start
loop_end:
\end{verbatim}

1. Control Flow Preservation: The structure of jump instructions preserves the original control flow.
2. Data Flow Preservation: The order of operations within and between blocks is maintained.
3. Composability: $T(basic\_for\_loop) \equiv T(B_{init}) \cdot T(B_{cond}) \cdot T(B_{body}) \cdot T(B_{incr})$, where $\cdot$ represents concatenation with appropriate jump instructions.

Therefore, the basic for loop is composable under $T$. Similar proofs can be constructed for other basic control structures.
\end{proof}

2. If-Else Statement:
Let $B_{cond}$, $B_{then}$, and $B_{else}$ be the basic blocks for condition, then-branch, and else-branch respectively.

Translation structure:
\begin{verbatim}
T(basic_if_else):
    T(B_cond)
    jz else_label
    T(B_then)
    jmp end_label
else_label:
    T(B_else)
end_label:
\end{verbatim}

Control flow and data flow preservation follow similarly to the for loop case.

3. While Loop:
Let $B_{cond}$ and $B_{body}$ be the basic blocks for condition and body respectively.

Translation structure:
\begin{verbatim}
T(basic_while):
loop_start:
    T(B_cond)
    jz loop_end
    T(B_body)
    jmp loop_start
loop_end:
\end{verbatim}

4. Do-While Loop:
Let $B_{body}$ and $B_{cond}$ be the basic blocks for body and condition respectively.

Translation structure:
\begin{verbatim}
T(basic_do_while):
loop_start:
    T(B_body)
    T(B_cond)
    jnz loop_start
\end{verbatim}

5. Switch-Case Statement:
Let $B_{expr}$ be the basic block for the switch expression, and $B_1, B_2, ..., B_n$ be the basic blocks for each case.

Translation structure:
\begin{verbatim}
T(basic_switch):
    T(B_expr)
    cmp result, case1_value
    je case1_label
    cmp result, case2_value
    je case2_label
    ...
    jmp default_label
case1_label:
    T(B_1)
    // No break implies fall-through
case2_label:
    T(B_2)
    ...
default_label:
    T(B_n)
end_switch:
\end{verbatim}

For all these structures, control flow is preserved by the appropriate use of jump instructions, and data flow is maintained by the sequential execution of basic blocks. The translation of each structure is a composition of its basic block translations, proving composability.

\begin{theorem}[Composability of Break and Continue Statements]
\label{thm:break_continue_composability}
Break and continue statements, which are basic statements as per \autoref{def:basic_statement}, are composable within their respective control structures when proper loop depth tracking is maintained.
\end{theorem}

\begin{proof}
Let $loop\_depth$ be a counter maintained during translation to track nested loop levels.

1. Break Statement:
Translation structure:
\begin{verbatim}
T(break):
    jmp loop_end_label_depth
\end{verbatim}
Where $loop\_end\_label\_depth$ corresponds to the end of the current loop at depth $loop\_depth$.

2. Continue Statement:
Translation structure:
\begin{verbatim}
T(continue):
    jmp loop_continue_label_depth
\end{verbatim}
Where $loop\_continue\_label\_depth$ corresponds to the continuation point of the current loop at depth $loop\_depth$.

Control flow is preserved by jumping to the appropriate label based on the current loop depth. Data flow is trivially preserved as these statements do not modify data.

The composability of these statements within their containing loops is maintained because:
a) They generate a single jump instruction that integrates with the loop's control flow.
b) The loop depth tracking ensures the jump targets the correct loop level in nested structures.
\end{proof}

\subsection{Composability of Complex Structures}

\begin{definition}[Composable Control Block]
\label{def:composable_control_block}
A composable control block is either:
\begin{itemize}
    \item A basic block as defined in \autoref{def:basic_block}, or
    \item A basic control structure as proved in \autoref{thm:basic_control_structure_composability}, or
    \item A sequence of composable control blocks, or
    \item A control structure whose all subparts are composable control blocks.
\end{itemize}
\end{definition}

\begin{figure}[t]
\small
\linespread{1.0}

\begin{algorithm}[H]
\caption{Iterative Bottom-Up Composability Proof Algorithm}
\label{alg:iterative_composability_proof}
\begin{algorithmic}[0]
\Procedure{ProveComposability}{Program $P$}
    \State $blocks \gets$ DecomposeIntoOutermostControlBlocks($P$) \Comment{Initial decomposition}
    \State $to\_process \gets$ new Deque()
    \For{each $block$ in $blocks$}
        \State $to\_process$.PushBack($block$) \Comment{Initialize processing queue}
    \EndFor
    \While{$to\_process$ is not empty}
        \State $current\_block \gets to\_process$.PopFront() \Comment{Handle first unhandled block}
        \If{IsBasicBlock($current\_block$)}
            \State $\textbf{continue}$ \Comment{Do nothing}
        \ElsIf{IsControlStructure($current\_block$)}
            \State $sub\_blocks \gets$ SplitControlStructure($current\_block$) 
            \For{each $sub\_block$ in $sub\_blocks$ in $reverse$ order}
                \State $to\_process$.PushFront($sub\_block$) \Comment{Handle sub-blocks in original order}
            \EndFor
        \Else
            \State \Return $P$ is not composable \Comment{Unrecognized structure}
        \EndIf
    \EndWhile
    \State \Return $P$ is composable
\EndProcedure

\Function{SplitControlStructure}{Block $b$}
    \If{$b$ is a For Loop}
        \State \Return SplitForLoop($b$)
    \ElsIf{$b$ is an If-Else structure}
        \State \Return SplitIfElse($b$)
    \Else
        \State \Return SplitOtherControlStructure($b$) \Comment{Extensible for other structures}
    \EndIf
\EndFunction

\Function{SplitForLoop}{ForLoop $f$} \Comment{Decompose for loop into constituent parts}
    \State \Return [ 
        $f.init$,
        $f$.ForBodyLabel,
        $f.cmp$,
        ConditionalJump($f$.ForEndLabel),\newline
        $f.body$,
        $f.incr$,
        UnconditionalJump($f$.ForBodyLabel),
        $f$.ForEndLabel
    ]
\EndFunction

\Function{SplitIfElse}{IfElse $i$} \Comment{Decompose if-else into constituent parts}
    \State \Return [ 
        $i.cmp$,
        ConditionalJump($i$.ElseLabel),
        $i.then\_body$,\newline
        UnconditionalJump($i$.EndIfLabel),
        $i$.ElseLabel,
        $i$.else\_body,
        $i$.EndIfLabel
    ]
\EndFunction

\end{algorithmic}
\end{algorithm}
\linespread{1.0}
\end{figure}

\begin{theorem}[Composability of Sequential Control Blocks]
\label{thm:sequential_composability}
A sequence of composable control blocks $CB_1, CB_2, ..., CB_n$ as defined in \autoref{def:composable_control_block} is composable under the translation function $T$.
\end{theorem}

\begin{proof}
Let $CB_1, CB_2, ..., CB_n$ be composable control blocks.
1. By \autoref{def:composable_control_block}, each $CB_i$ is composable.
2. Translation structure: $T(CB_1 \circ CB_2 \circ ... \circ CB_n) \equiv T(CB_1) \cdot T(CB_2) \cdot ... \cdot T(CB_n)$
   where $\circ$ denotes sequential composition in SRC and $\cdot$ denotes concatenation in DST.
3. Control Flow Preservation: The sequential order of control blocks is maintained in the translation.
4. Data Flow Preservation: The order of operations between control blocks is preserved.

Therefore, the sequence of composable control blocks is itself a composable control block under $T$.
\end{proof}

\begin{theorem}[Composability of Arbitrary Programs]
\label{thm:arbitrary_program_composability}
Any program $P$ that can be decomposed into a sequence of control blocks as defined in \autoref{def:control_block} is composable under the translation function $T$ if the Iterative Composability Proof algorithm (\autoref{alg:iterative_composability_proof}) marks it as composable.
\end{theorem}

\begin{proof}
The proof follows from the correctness of the Iterative Composability Proof algorithm:
\begin{enumerate}
    \item The algorithm starts with basic blocks and basic control structures, which are proven composable by \autoref{thm:basic_statement_composability} and \autoref{thm:basic_control_structure_composability}.
    \item It iteratively builds up composability for larger structures:
        \begin{itemize}
            \item Sequences of composable blocks are proved composable by \autoref{thm:sequential_composability}.
            \item Control structures with all composable subparts are marked composable.
        \end{itemize}
    \item The process continues until the entire program is marked composable or no further progress can be made.
    \item If the entire program is marked composable, it means that $T(P)$ can be expressed as a composition of the translations of its composable parts, preserving both control flow and data flow as per \autoref{def:translation_composability}.
\end{enumerate}

Therefore, if the algorithm returns that $P$ is composable, then $P$ is indeed composable under the translation function $T$.
\end{proof}

\begin{theorem}[Composability of Goto Statements]
Goto statements, which are basic statements as per \autoref{def:basic_statement}, are composable under the translation function $T$, but aribitrary goto statements can break the structured control flow assumed in the main proof.
\end{theorem}

\begin{proof}
Let $l$ be a label and $goto\ l$ be a goto statement.

Translation structure:
\begin{verbatim}
T(goto l):
    jmp label_l

T(l:):
label_l:
\end{verbatim}

The goto statement translates to an unconditional jump, preserving control flow. It doesn't directly affect data flow. Composability holds as $T(stmt_1 \circ goto\ l \circ stmt_2) \equiv T(stmt_1) \cdot T(goto\ l) \cdot T(stmt_2)$.

However, goto introduces complications:
\begin{itemize}
    \item Non-local control flow can break the nested structure of control blocks.
    \item Programs with unrestricted goto usage are difficult to decompose into well-defined control blocks.
    \item It can lead to unstructured code, complicating reasoning about program behavior.
\end{itemize}
\end{proof}

While goto is provably composable, it's discouraged in modern programming for readability, maintainability, and optimization reasons. Our composability principle is most applicable and valuable in the context of structured programming paradigms.

\subsection{Scope and Limitations of the Proof}

The proof of composability presented in this paper is based on a simplified model of C-like languages and unoptimized translation. It's important to note several key points about the scope and limitations of this proof:

\begin{enumerate}
    \item \textbf{Simplification and Correctness:} The simplifications made in our language model and translation process do not compromise the validity of the proof. The core of our argument relies on the decomposition of programs into control blocks and the composability of these blocks. The internal structure of basic blocks, while important for actual compilation, does not affect the composability principle we've established.

    \item \textbf{Unoptimized Translation:} Our proof assumes a straightforward, unoptimized translation process. This assumption is crucial for maintaining the direct correspondence between source code structures and their translations.

    \item \textbf{Limitations for Complex Language Features:} The composability principle as proved here can be applied to C-like languages, but may not hold for more complex language features. For example:
    \begin{itemize}
        \item Exception Handling: Languages with sophisticated exception handling mechanisms, such as Python, introduce complexities that can break composability. These mechanisms often require:
        \begin{itemize}
            \item Guarded execution of code blocks.
            \item Runtime type information (RTTI) for determining appropriate exception handlers.
            \item Non-local control flow that can't be easily decomposed into our model of control blocks.
        \end{itemize}
        
        \item Coroutines and Generators: Features that allow for suspending and resuming execution mid-function can introduce state that is not easily captured in our model of control flow.
        
        \item Reflection and Metaprogramming: Languages that allow for runtime modification of program structure or behavior can invalidate static composability assumptions.
    \end{itemize}

    Although not applicable to some specific language features, it doesn't mean the composablity and its derived LEGO translation method is not applicable to the whole programming language, as long as these features are not used in the code, the composability will still stand and the LEGO translation will still work.

    \item \textbf{Optimizations Across Basic Blocks:} Our proof assumes that the boundaries of control blocks are respected in the translation process. However, many real-world compiler optimizations operate across these boundaries. Examples include:
    \begin{itemize}
        \item Loop unrolling
        \item Function inlining
        \item Global value numbering
        \item Code motion optimizations
    \end{itemize}
    Such optimizations can reorder, eliminate, or combine operations from different control blocks, potentially breaking the composability property as we've defined it. However, some local optimizations, like mem2reg and strength reduction can still be performed, and is observed through our evaluation where the model tends to perform such optimizations.

    \item \textbf{Applicability:} Despite these limitations, the composability principle proved here is valuable for:
    \begin{itemize}
        \item The foundation of LEGO translation method, the proof reveals the composable nature of code in at least control block level, which is a major difference than natural languages.
        \item The proof process also guided \autoref{algo_part_split} in LEGO translation, as proving the composability and making use of the composability share similar algorithms.
    \end{itemize}
\end{enumerate}

In conclusion, while our proof provides a strong foundation for understanding composability in C-like languages with straightforward translation, it's important to recognize its boundaries. More complex language features may require extensions or modifications to this framework to maintain composability guarantees. And optimized code translation usually is not composable.

\section{Discussions}
\label{discussion}

\subsection{Universality of LEGO translation}

The LEGO translation method, while initially developed for compilation tasks, demonstrates broader applicability based on fundamental properties of programming languages rather than being specific to compilation. The composability that LEGO translation leverages stems from the well-encapsulated control flow and locality principles inherent in modern programming languages (disregarding constructs like \textbf{goto} in C, more limitations are clearly described in \autoref{limitation}).

These characteristics are intrinsic to programming languages themselves and have guided modern compiler design. They enable the modular partitioning of large-scale programs in modern software development, allowing for incremental and even parallel compilation of code. We harness these properties and apply them to the context of neural compilation using Large Language Models (LLMs).

It's important to note that the applicability of LEGO translation extends beyond compilation. It is suitable for various tasks originating from programming languages, such as code translation between different languages. This method significantly enhances the scalability of machine translation tasks for code, providing a powerful tool for handling large and complex codebases.

\subsection{Managing Highly Complex Expressions}

One of the primary challenges in neural compilation arises when dealing with expressions or statements of high complexity. In such cases, LLMs struggle to accurately evaluate these expressions through next token prediction. To address this, we propose two solutions:
\begin{itemize}
    \item External Tool Integration: We can utilize external parsing tools to generate tree structure information for complex expressions evaluation. This tree structure is then provided to the LLM, offering an explicit traversal order and guiding the evaluation process.
    \item Expression Decomposition: Without relying on external tools, we can design a new pass where the LLM identifies high-complexity expressions and rewrites them as a combination of lower-complexity expressions. This approach ensures that the entire program consists only of expressions within a proper LLM's evaluation capabilities.
\end{itemize}

\subsection{Computational Cost, Effectiveness, and Future Prospects}

While our neural compilation method is primarily a proof of concept, it does incur significantly higher computational costs compared to traditional compilation methods - approximately $10^6$ to $10^7$ times higher. However, this should be weighed against the substantial human resources required for traditional compiler development.

The key advantage of our approach lies in its potential for rapid adaptation to new instruction set extensions or frontend intrinsics. Through techniques like RAG (Retrieval-Augmented Generation) and in-context learning, our method can be extended to support new architectures or language features. This positions neural compilation as a valuable assistant in the compiler development process. A particularly promising application is in generating end-to-end unit tests for compiler adaptation to new instructions. This could significantly streamline the development and testing phases of compiler updates.
Recent research like ~\cite{munley2024llm4vv} has shown the ability to use LLMs to generate unit tests during compiler validations. 
\section{Evaluation Details}
This section provides more details figures, tables and further explanations about \textbf{LEGO-Compiler} design and experiment evaluation.

\subsection{LEGO-Compiler: detailed designs}
\label{lego_compiler_more}

{
  \begin{algorithm}[t!]
    \caption{LLM-driven \textbf{Part Split} Algorithm based on Control Blocks}
    \label{algo_part_split}
    \begin{algorithmic}
\Procedure{SelectControlBlocks}{$function$}
\State $blocks \gets \emptyset$
\State $deque.push\_back ({function})$
\While{$deque$ is not empty}
\State $block \gets deque.pop\_front()$
\State $decision \gets$ LLMDecideSplit($block$)
\If{$decision$ is "keep"}
\State $blocks$.append($block$)
\Else
\State $subBlocks \gets$ SplitByOutermostControl($block$)
\For{$subBlock$ in $subBlocks$ in $reverse$ order}
\State $deque.push\_front$($subBlock$)
\EndFor
\EndIf
\EndWhile
\State \Return $blocks$
\EndProcedure
    \end{algorithmic}
  \end{algorithm}
}

\begin{figure}
\begin{center}
\includegraphics[width=0.98\textwidth]{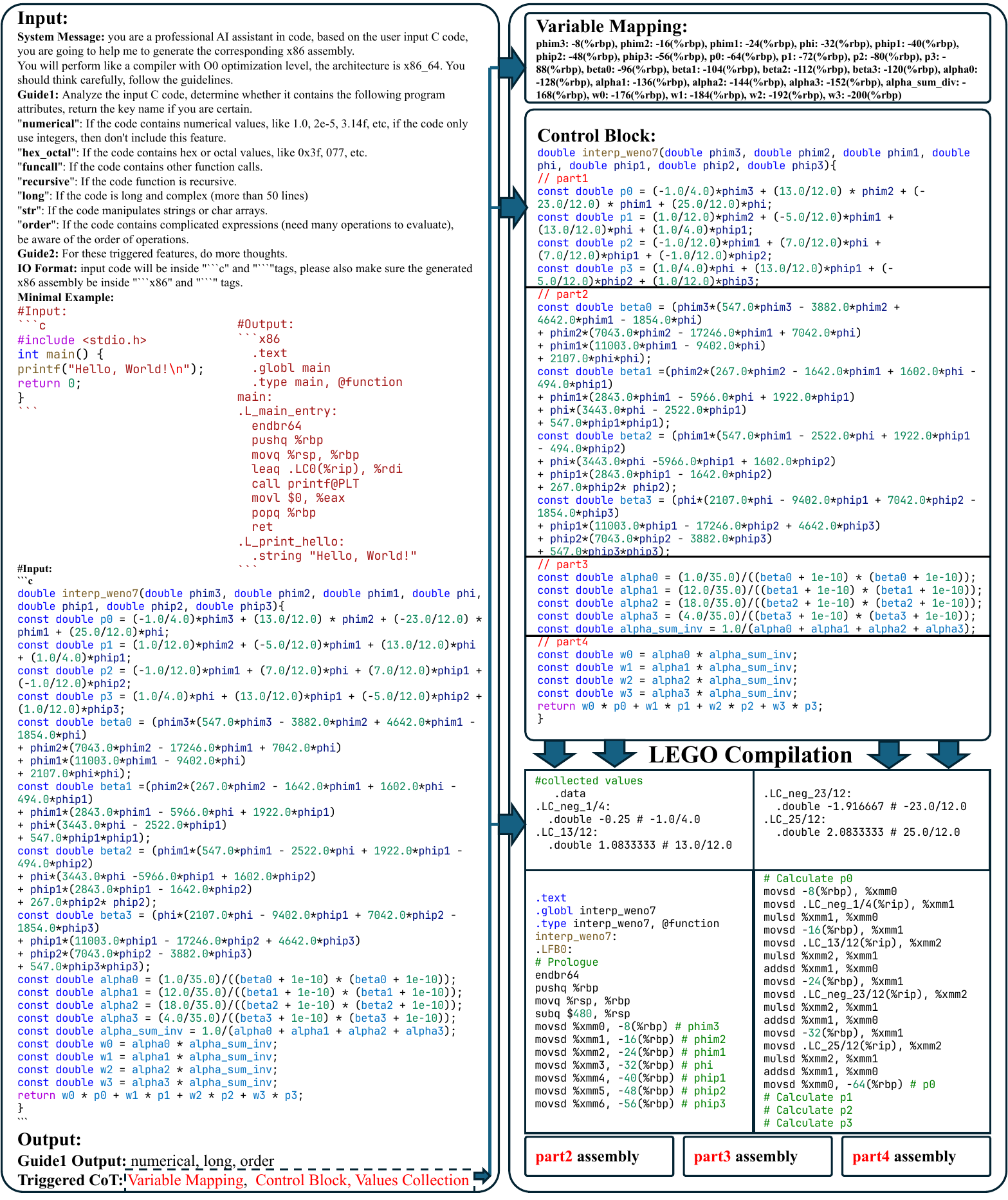}
\end{center}
\caption{Example workflow for LEGO-Compiler on a full ExeBench example: source code analysis triggers thoughts, including \textbf{variable mapping}, splitting \textbf{control blocks} and \textbf{value collection} illustrated.}
\label{fig_overview}
\end{figure}

As depicted in \autoref{fig_overview}, LEGO-Compiler is designed to perform a series of steps guided by compiler expert knowledge, just like Chain-of-Thoughts(CoTs). However, not all CoTs are necessary for each input code, so in our design, we have an \textbf{analyze-then-think} approach.
First, we will perform an analyzing pass to scan the whole program, whose output flags would trigger necessary Chain-of-Thoughts that will be used in the following process. In this example, the code pattern is majorly about double-precision floating point calculations (\textbf{numerical}) and complicated expression evaluation (\textbf{order}), besides, the code is too long for direct translation method to handle (\textbf{long}). Thus, based on the analysis, we applied the following CoTs: 
\begin{itemize}
    \item \textbf{Values collection}: A necessary thought, collecting all variables, numericals in a scanning pass, the \textbf{numerical} flag will teach the LLM about assembly knowledge to save numerical values.
    \item \textbf{Variable mapping}: Another necessary thought, which will base on the scanned variables and their types, and form a variable mapping table (SymbolTable) for later compilation.
    \item \textbf{Control Block}: the LEGO translation methodology is applied triggered by \textbf{long}, where the entire code is considered too long and will be split into control-block level code snippets via \autoref{algo_part_split}, it's noteworthy that the \textbf{order} flag from analysis will suggest the LLM to split the program into finer-grained blocks so that they can focus more on the order of operations within each block, in \autoref{fig_overview}, there is just one basic block, the flag suggests LLM to split into 4 sequential parts. Then these parts are translated with the aid of SymbolTable individually. Finally, these compiled results are composed together to form a full LEGO compilation.
\end{itemize}  

Therefore, we apply a \textit{residual} step-by-step workflow. With different input code, the triggered CoTs will be different, and non-triggered CoTs are just skipped with a residual connection.  
\subsection{ExeBench breakdown}
\label{breakdown_exebench}

\begin{figure}[ht]
    \centering
    \begin{subfigure}{0.32\textwidth}
        \centering
        \includegraphics[width=\textwidth]{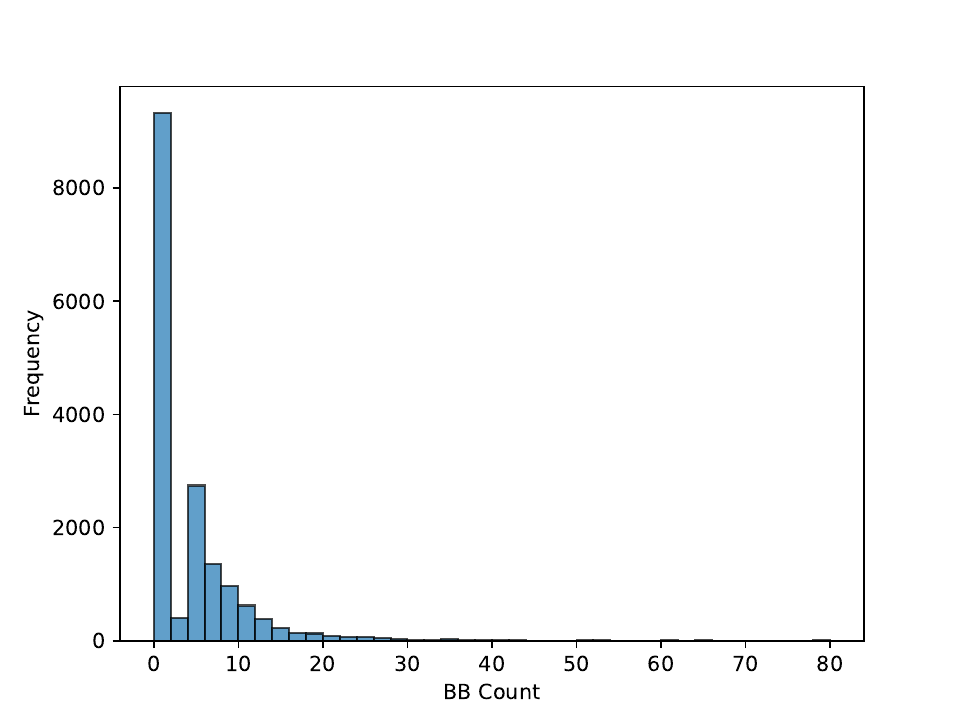}
        \label{fig:sub1}
    \end{subfigure}
    \hfill
    \begin{subfigure}{0.32\textwidth}
        \centering
        \includegraphics[width=\textwidth]{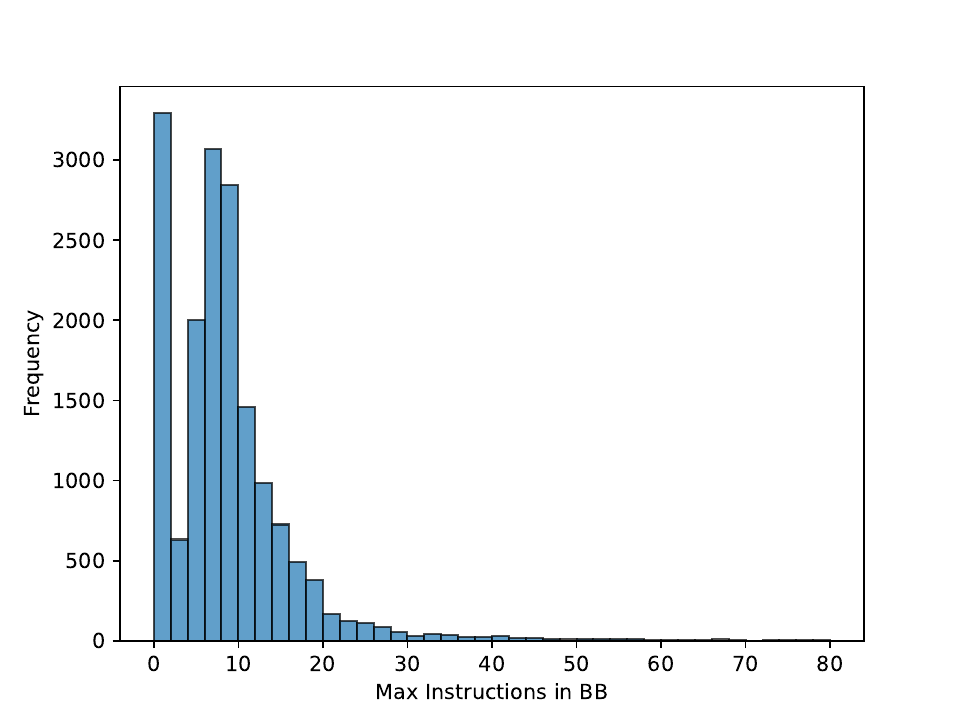}
        \label{fig:sub2}
    \end{subfigure}
    \hfill
    \begin{subfigure}{0.32\textwidth}
        \centering
        \includegraphics[width=\textwidth]{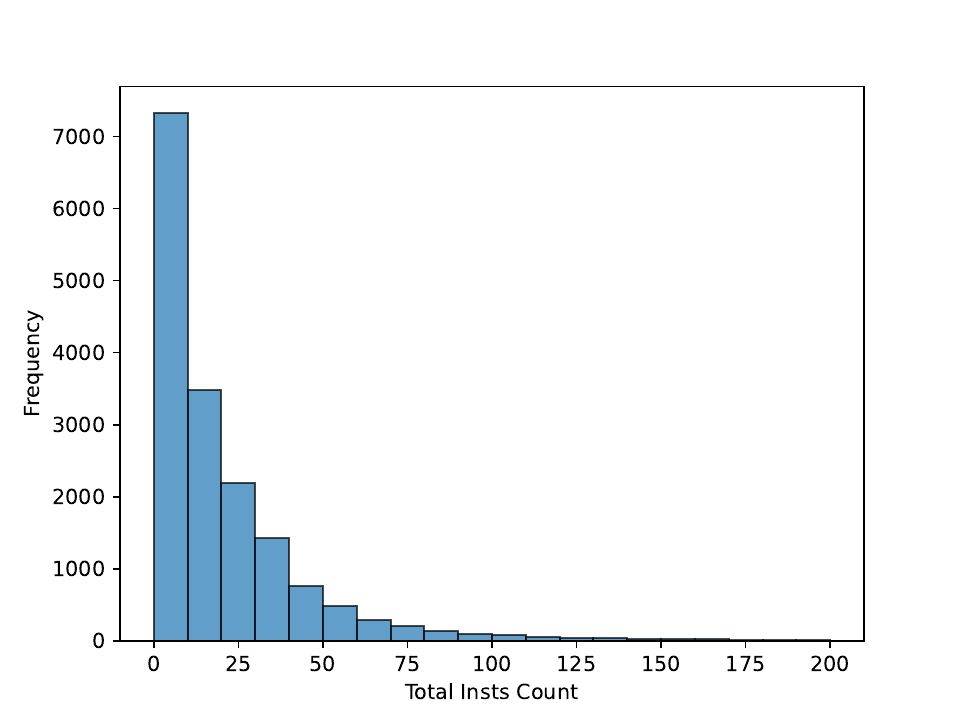}
        \label{fig:sub3}
    \end{subfigure}

    \begin{subfigure}{0.32\textwidth}
        \centering
        \includegraphics[width=\textwidth]{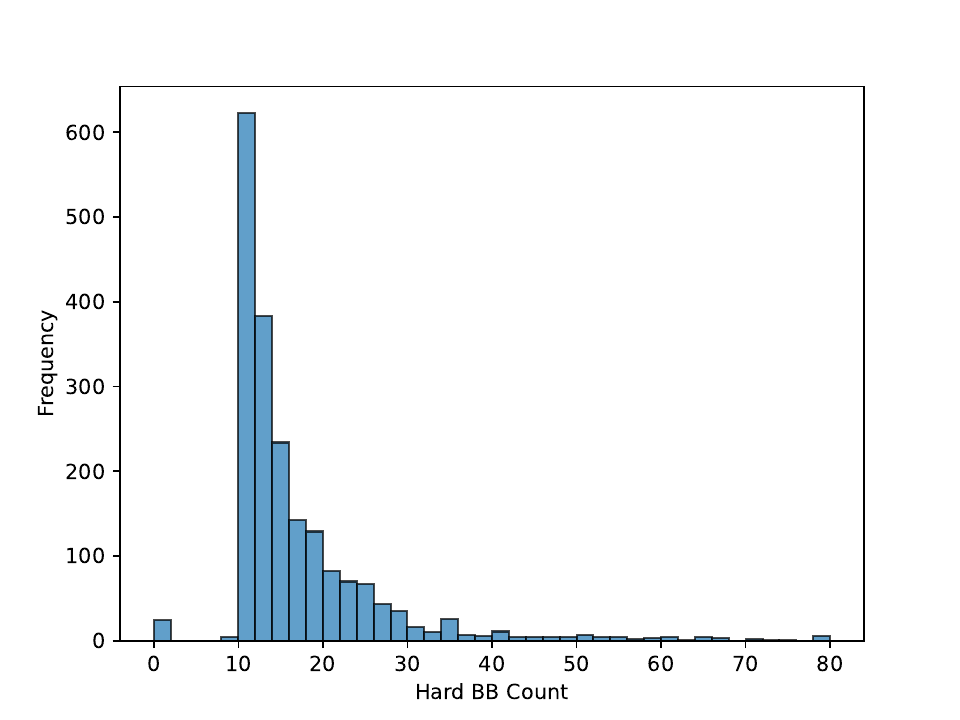}
        \label{fig:sub4}
    \end{subfigure}
    \hfill
    \begin{subfigure}{0.32\textwidth}
        \centering
        \includegraphics[width=\textwidth]{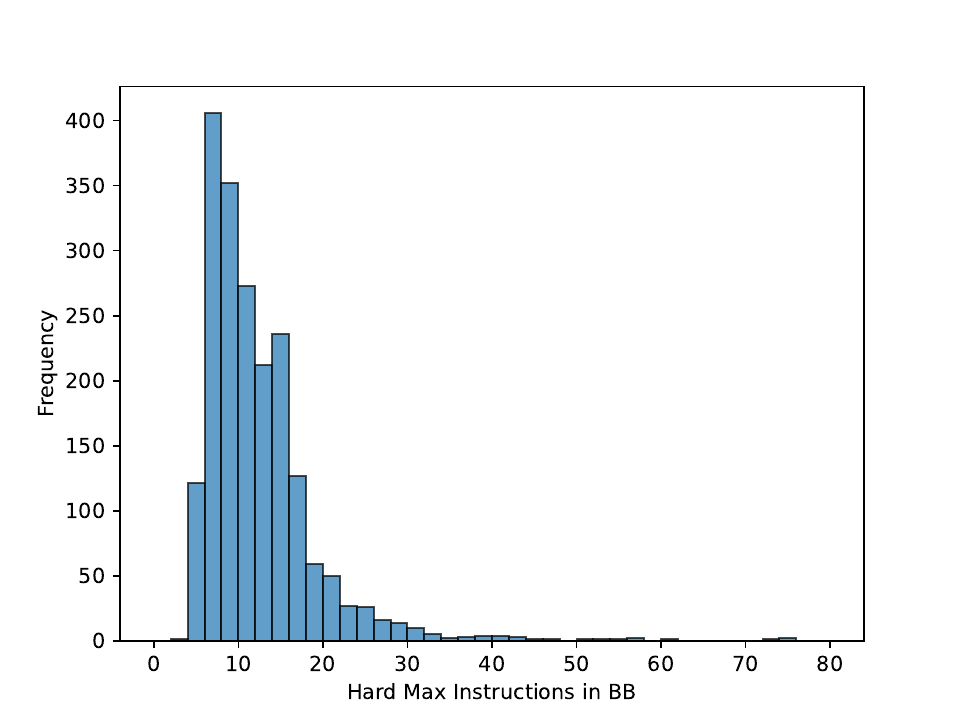}
        \label{fig:sub5}
    \end{subfigure}
    \hfill
    \begin{subfigure}{0.32\textwidth}
        \centering
        \includegraphics[width=\textwidth]{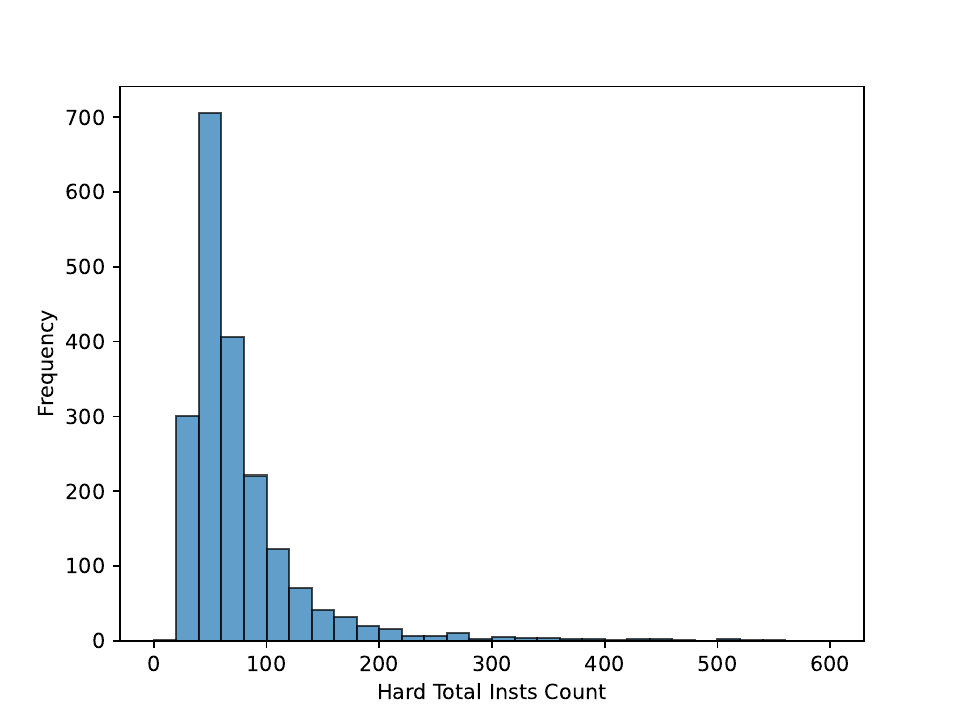}
        \label{fig:sub6}
    \end{subfigure}

    \caption{Complexity breakdown of ExeBench and its hard 10\% (roughly) subset, we use llvm as the analysis tool, then filter the subset with the following conditions: number of basic blocks(BB) $\geq$ 10 or max instructions in BB $\geq$ 80 or total instructions $\geq$ 200. Upper figures characterize the overall of Exebench and Lower figures characterize the hard 10\% subset.}
    \label{fig:exebench_breakdown_grid_view}
\end{figure}

The failure cases observed in ExeBench can be majorly categorized into three types: 
    \begin{itemize}
        \item The insufficiency on some language-specific features, for example, lacking the knowledge of certain operations, which can be definitely improved with more data in the next model pretrained or by providing external knowledge to aid its generation. 
        \item The unsuccessful reasoning step in the workflow. This method requires the LLMs to reason arithmetic computation and capture specific code patterns in the code to form intermediate results to aid the generation. If the reasoning process generates incorrectly (rare), the whole workflow will fail. However, the reasoning capabilities required for this method is not high, majorly the addition and multiplication of integer values within 1000(typically). The error-feedback loop can mitigate most of trivial errors during reasoning. Besides, as LLMs keep improving their abilities in reasoning and math, this type of failures will reduce significantly.
        \item Very long code reasoning and follow-up generation, where LLMs fail to generate a very large output at once. The first reason is the limitation of current LLMs themselves, although advanced LLMs have increased their context limits into hundreads of thousands tokens, their single generation capability is still limited, to either 4096, 8192 or 16384 tokens. The second reason is the difficulty to generate a long, error-prone output(like assembly languages) at once, this is an intrinsic drawback of direct generation method itself, and can be improved greatly with the proposed LEGO translation/compilation method. LEGO translation can reduce the complexity to control block level, or at maximum, statement level, however, if the statement itself is very long and complicated to evaluate (which is very rare during evaluation, but potential in modern programming paradigms), our method doesn't tackle it, which is a limitation in our work. However, more program-rewriting steps can help to mitigate such issues and is verifiable, and we leave it as future work.
    \end{itemize}

\begin{figure}
    \begin{center}
    \includegraphics[width=0.95\textwidth]{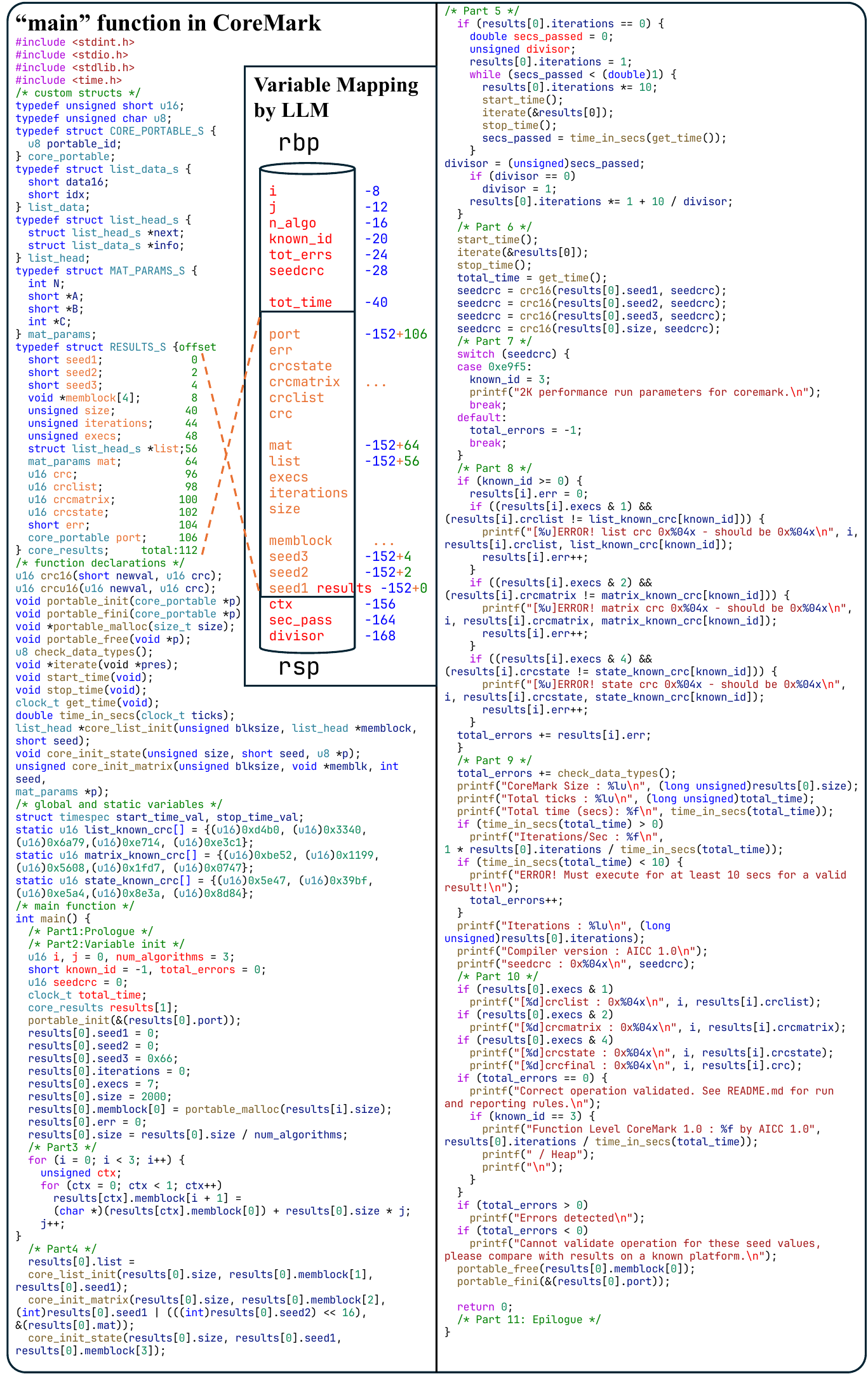}
    \end{center}
    \caption{The CoreMark main function, one of the most difficult code we evaluated. In this figure, all CoTs are illustrated in the code annotations in color, as well as the variable mapping process.}
    \label{coremark_main}
\end{figure}

\begin{figure}
\begin{center}
\includegraphics[width=0.98\textwidth]{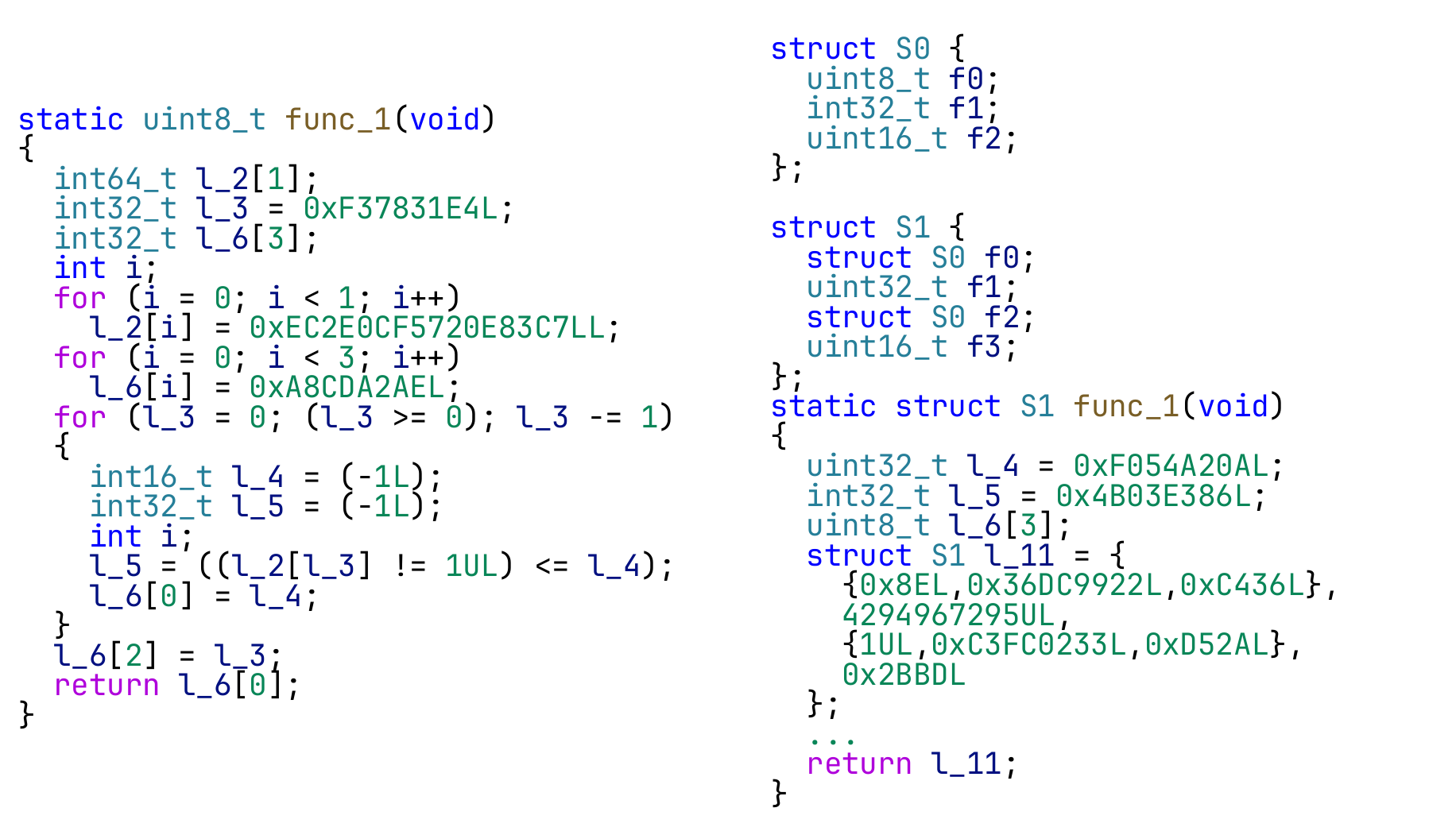}
\end{center}
\caption{Csmith example code, the major body part of the right hand side code is omitted. This example characterizes the necessity of both the Chain-of-Thought reasoning of structs and stack allocation and the LEGO translation method to overcome the complexity of coarse-grained translation.}
\label{fig_csmith_example1}
\end{figure}

\begin{figure}
\begin{center}
\includegraphics[width=0.98\textwidth]{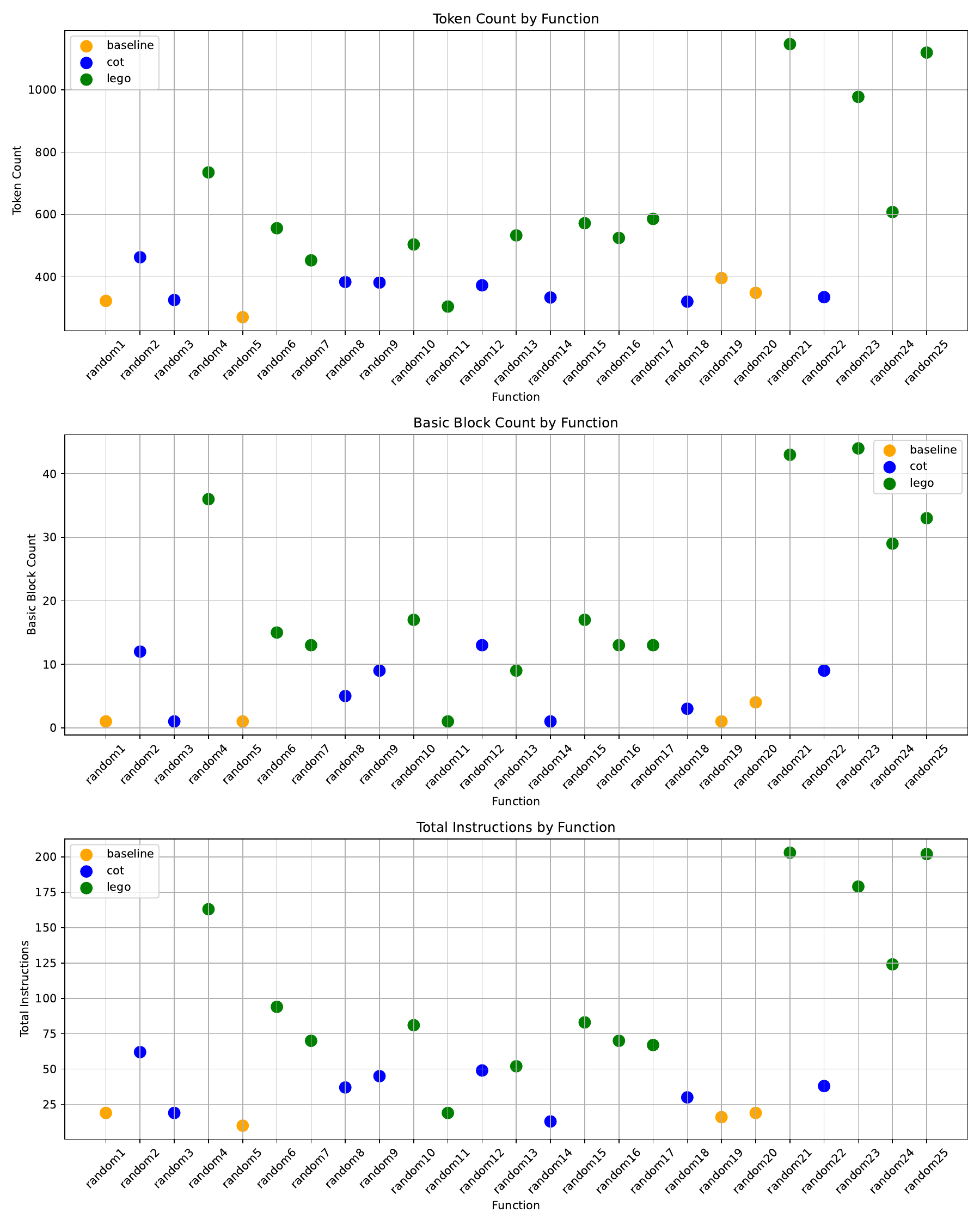}
\end{center}
\caption{Csmith random generated code statistics, where the practical utility of the LEGO method is show clearly by passing significantly more complex cases.}
\label{fig_csmith_results}
\end{figure}

\subsection{Other evaluation details}

\autoref{tab:results_temperature} shows the impact of temperature when using LLMs for neural compilation. LLMs have better Pass@1 accuracy when temperature is low,  but higher Pass@5 accuracy when temperature is high. This is as expected, since temperature influences the decoding process, with higher temperature, the results are more diverse, allowing LLMs to jump out of pretraining bias, however, this could also cause more errors by choosing sub-optimal decoding tokens that may cause errors. 

\begin{table}[t]

  \caption{Ablation study: impact of temperature on Pass@1 and Pass@5 performance}
  \label{tab:results_temperature}

  \centering
  \resizebox{\textwidth}{!}{%
    \begin{tabular}{lcccccccccc}
      \toprule
      \multirow{2}{*}{Model} & \multicolumn{5}{c}{Pass@1} & \multicolumn{5}{c}{Pass@5} \\
      \cmidrule(lr){2-6} \cmidrule(lr){7-11}
      & 0.2 & 0.4 & 0.6 & 0.8 & 1.0 & 0.2 & 0.4 & 0.6 & 0.8 & 1.0 \\
      \midrule
      GPT-4o & 71\% & \textbf{73\%} & 72\% & 72\% & 72\% & 79\% & 83\% & 86\% & 89\% & \textbf{92\%} \\
      Claude-3.5-Sonnet & 87\% & 91\% & \textbf{93\%} & 88\% & 89\% & 91\% & 92\% & \textbf{96\%} & 94\% & \textbf{96\%} \\
      DeepseekCoder & \textbf{89\%} & 88\% & 86\% & 87\% & 88\% & 92\% & 92\% & 92\% & \textbf{93\%} & 92\% \\
      \midrule
      GPT-4o-mini & \textbf{64\%} & 61\% & 61\% & 60\% & 60\% & 71\% & 71\% & 79\% & 73\% & \textbf{80\%} \\
      Claude-3-Haiku & \textbf{79\%} & 76\% & 78\% & 72\% & 73\% & 82\% & 84\% & 85\% & \textbf{86\%} & \textbf{86\%} \\
      \bottomrule
    \end{tabular}%
  }
\end{table}

\section{Limitations}
\label{limitation}

\textbf{Optimization Capabilities}: The major focus of LEGO-Compiler is on the translation correctness rather than code optimization. Traditional compilers excel at producing highly optimized code, a capability not yet matched by our neural approach. Although we do find LLMs are capable to optimize the translation process individually and generate optimized assembly code, it is not our main focus.
Future work could explore integrating optimization techniques into the neural compilation process.

\textbf{Performance Overhead}: As noted in the discussion, the computational cost of neural compilation is significantly higher than traditional methods. This limitation may restrict its practical application in scenarios where compilation speed is critical.

\textbf{Complex Expression Handling}: The paper acknowledges challenges in managing highly complex expressions, proposing external tool integration or expression decomposition as potential solutions. This indicates a current limitation in LLMs' ability to handle intricate code structures independently.

\textbf{Architecture-Specific Knowledge}: While the paper demonstrates success with x86, ARM, and RISC-V architectures, expanding to a broader range of architectures, especially more specialized ones, may require significant additional training or fine-tuning of the LLMs or providing large RAG database to provide such knowledge in the context.

\textbf{Security and Reliability}: The stochastic nature of LLM outputs raises concerns about the consistency and security of the generated assembly code. Ensuring deterministic outputs and preventing potential vulnerabilities introduced by the neural compilation process remains a challenge.

\textbf{Handling of Language-Specific Features}: The paper primarily focuses on C-like language compilation and proves the availability of functionality in neural compilation through both theoretical and empirical results. However, extending the approach to other programming languages can result in more tailored problems, for example:
\begin{itemize}
    \item \textbf{RAII idiom}: Languages with class properties, like C++, have an important programming idiom called \textbf{R}esource \textbf{A}cquisition \textbf{I}s \textbf{I}nitialization(\textbf{RAII}), which pose significant challenges for LLMs. For instance, constructor and destructor functions in these languages are implicitly called based on scope. This implicit behavior is difficult for LLMs to accurately model and implement in assembly code.
    \item \textbf{Name Mangling}:Languages like C++ and Rust use name mangling mechanisms for function overloading and template instantiation. This requires special handling of global symbols, such as function names, during compilation, which may be challenging for LLMs to consistently implement without explicit training on these concepts, but this could be solved by using external mangling tools like\textbf{ c++filt}~\cite{cppfilt}.
    \item \textbf{Dynamic Language Features}: Some language features violate the composability principle that LEGO translation relies on. For example, Python's exception handling mechanism, which can cross scope boundaries, would make the LEGO translation method ineffective for such features.
\end{itemize}

It's important to note that many of these challenges are not unique to neural compilation. Traditional compilers also struggle with highly dynamic features like exception handling and Run-Time Type Information (RTTI). Languages like Python achieve flexibility by sacrificing native code generation in favour of interpretation or JIT compilation. Therefore, these limitations are not specific to our work but rather inherent to any approach based on static compilation analysis.

The ability to handle these diverse language features represents an area for future research in neural compilation. It may require developing specialized techniques or combining neural methods with traditional compiler approaches to address these complex language-specific challenges.

Scalability to Very Large Codebases: While the LEGO translation method significantly improves scalability, handling entire large-scale software projects or operating systems may still be beyond the current capabilities of this approach. However, It is noteworthy that repository complexity is naturally reduced into files or functions, therefore, LLM-based compilers and translators are potential to translate them with more advanced models and more carefully designed methods.

\end{document}